\documentclass[a4paper,10pt,3p,twocolumn]{elsarticle}

\usepackage{lineno,hyperref}
\modulolinenumbers[5]
\usepackage{amssymb, amsmath, MnSymbol}
\usepackage{stmaryrd}
\usepackage{graphicx}

\usepackage{microtype}

\usepackage{booktabs} 

\usepackage[linesnumbered, lined, boxed]{algorithm2e}

\usepackage{todonotes}

\usepackage{url}

\usepackage{mathtools}

\usepackage{microtype}

\usepackage[capitalise,nameinlink]{cleveref}

\usepackage{csquotes} 

\usepackage{paralist}
\usepackage{mathtools}

\usepackage{fontawesome}

\usepackage[edges]{forest}
\usetikzlibrary{arrows.meta}

\forestset{%
  colour me out/.style={outer color=#1!75, inner color=#1!75, draw=darkgray, thick, rounded corners},
  rect/.append style={rectangle, rounded corners=2pt},
  dir tree switch/.style args={at #1}{%
    for tree={
      edge=-Latex,
      fit=rectangle,
    },
    where level=#1{
      for tree={
        folder,
        grow'=0,
      },
      delay={child anchor=north},
    }{},
    before typesetting nodes={
      for tree={
        content/.wrap value={\strut ##1},
      },
      if={isodd(n_children("!r"))}{
        for nodewalk/.wrap pgfmath arg={{fake=r,n=##1}{calign with current edge}}{int((n_children("!r")+1)/2)},
      }{},
    },
  },
}
\usetikzlibrary{tikzmark,calc,decorations.pathreplacing} 

\usepackage{caption}
\usepackage{subcaption}
\usepackage{wrapfig}

\usepackage{multirow}
\usepackage{threeparttable}

\usepackage{amsthm}
\usepackage{microtype}

\usepackage[inline]{enumitem}

\usepackage{scrhack}                    
\usepackage{listings}                   

\newcommand{\qedwhite}{\hfill \ensuremath{\Box}}

\newcommand{\labeltitle}[1]{\vskip 0.03in \noindent\textbf{#1}} 
\newcommand{\labelsubtitle}[1]{\vskip 0.03in \noindent\emph{#1}.}

\newcommand{\eg}{e.g.,\ }
\newcommand{\ie}{i.e.,\ }

\newcommand{\Bool}{\mathsf{B}}

\newcommand{\BExp}{\mathsf{BExp}}
\newcommand{\Exp}{\mathsf{Exp}}

\newcommand{\mathvar}[1]{\mathord{\mathit{#1}}}

\newcommand{\type}{\mathvar{type}}

\newcommand{\pchar}{\mathvar{char}}
\newcommand{\PC}{\mathvar{PC}}

\newcommand{\CPT}{\mathvar{CPT}}
\newcommand{\EL}{\mathvar{EL}}

\newcommand{\iC}{\mathvar{inContr}}
\newcommand{\oC}{\mathvar{outContr}}
\newcommand{\contr}{\mathvar{contr}}

\usepackage{xspace}

\newcommand{\tdbnet}{timed DB-net\xspace}
\newcommand{\tdbnets}{timed DB-nets\xspace}

\newcommand{\match}{\mathsf{match}}

\newcommand{\qedblack}{\hfill \ensuremath{\blacksquare}}

\usepackage{mdframed}
\newmdtheoremenv{insight}{Conclusions}

\theoremstyle{definition}
\newtheorem{example}{Example}
\newtheorem{definition}{Definition}
\newtheorem{theorem}{Theorem}

\newtheorem{lemma}[theorem]{Lemma}

\definecolor{amber}{rgb}{1.0,0.75,0.0}

\tikzstyle{vertex}=[circle, draw, inner sep=0pt, minimum size=6pt]

\usepackage[capitalise,nameinlink]{cleveref}
\crefname{section}{Sect.}{Sects.}
\Crefname{section}{Section}{Sections}
\crefname{figure}{Fig.}{Figs.}
\Crefname{figure}{Figure}{Figures}
\crefname{table}{Tab.}{Tabs.}
\Crefname{table}{Table}{Tables}
\crefname{lstlisting}{List.}{Lists.}
\Crefname{lstlisting}{Listing}{Listings}
\crefname{example}{Ex.}{Exs.}
\Crefname{example}{Example}{Examples}
\crefname{theorem}{Lem.}{Lemmas}
\Crefname{theorem}{Lemma}{Lemmas}

\usepackage{listings} 

\mathchardef\UrlBreakPenalty=10000

\begin{document}

\begin{frontmatter}

\title{Cost-aware Process Modeling in Multiclouds}

\author[add1]{Daniel Ritter}{\corref{mycorrespondingauthor}}
\ead{daniel.ritter@sap.com}

\address[add1]{SAP, Cloud Platform, Walldorf, Germany}

\begin{abstract}
Integration as a service (INTaaS) is the centrepiece of current corporate, cloud and device integration processes.
Thereby, processes of integration patterns denote the required integration logic as integration processes, currently running in single-clouds.
While multicloud settings gain importance, their promised freedom of selecting the best option for a specific problem is currently not realized as well as security constraints are handled in a cost-intensive manner for the INTaaS vendors (\eg by isolating tenants on container level), leading to security vs. costs goal conflicts, and intransparent to the process modeler.

In this work, we propose a design-time placement for processes in multiclouds that is cost-optimal for INTaaS problem sizes, and respects configurable security constraints of their customers.
To make the solution tractable for larger, productive INTaaS processes, it is relaxed by using a local search heuristic, and complemented by correctness-preserving model decomposition.
This allows for a novel perspective on cost-aware process modeling from a process modeler's perspective.

The multicloud process placement is evaluated on real-world integration processes with respect to cost- and runtime-efficiency, resulting in interesting trade-offs.
The process modeler's perspective is investigated based on a new cost-aware modeling process, featuring the interaction between the user and the INTaaS vendor through ad-hoc multicloud cost calculation and correctness-preserving, process cost reduction proposals.
\end{abstract}

\end{frontmatter}

\section{Introduction}
%

Recent trends like cloud, mobile, serverless computing result in a growing number of communication partners as well as more complex integration processes that have to be operated by Integration as a Service (INTaaS) systems~\cite{Ritter201736}.
The integration processes essentially denote compositions of integration patterns 
\cite{Ritter201736,DBLP:conf/debs/0001MFR18,hohpe2004enterprise}, connecting the different communication endpoints, which are developed by users / customers (called tenants) and operated by INTaaS vendors.
While INTaaS vendors, like SAP Cloud Platform Integration (CPI)~\cite{sap-hci-content}, should strive for cost-optimal, multicloud platform deployments (\ie one process potentially runs across multiple platforms, \eg Amazon Web Services (AWS), Microsoft Azure)~\cite{buyya2010intercloud}, the customers or tenants impose strict security (\eg \emph{EU Data Protection Regulation}\footnote{EU --- GDPR, visited 2/2021: \url{http://goo.gl/Ru0slz}.}: non-shareable process parts)~\cite{singhal2013collaboration} and contractual requirements on these vendors (\eg exemption of one platform like \enquote{not on AWS}).
This results in goal-conflicts between the diverse requirements of the INTaaS customers and the costs of the vendors operating their (private or public) solutions on different cloud platforms (cf. \cref{ex:motivation}).
\begin{example} \label{ex:motivation}
	\Cref{tab:example_capacity} shows three tenants $\{t_1,t_2,t_3\}$, representing customers of an INTaaS, with processes (\texttt{$pc_{t_1,1}$}, \texttt{$pc_{t_1,2}$}, \texttt{$pc_{t_1,3}$}, \texttt{$pc_{t_2,1}$}, \texttt{$pc_{t_3,1}$}) that have a security
	constraint called \textsc{SH} (\ie shareable: not violating the confidentiality of processes of other tenants) and have required resource capacities (\textsc{CAP}).
	Shareable integration processes (marked with \faThumbsOUp) are free of side-effects and could be operated in a shared, multi-tenant way.
	In contrast, the non-shareable processes (marked with \faThumbsDown) can have side-effects, and thus need to be isolated.
	
	A classical hosting approach shown in~\cref{fig:hosting} leaves the customers' processes on a dedicated, private environment ($c_1$,..,$c_3$), fulfilling the security requirements (non-shareable processes marked with \textsc{N}), but imposing high costs on the INTaaS vendors (between $60.00$--$90.00$ EUR/mo\footnote{Considering example platform vendor variants $1$ and $N$, $\{cpu=1, ram=6.25GB, cost/mo=30$ EUR$\}$ and $\{cpu=1, ram=3.125GB, cost/mo=20$ EUR$\}$, respectively.}).
	A security-compliant but cost-efficient alternative ($50$ EUR/mo) would be a placement of processes $\{pc_{t_1,1}, pc_{t_2,1}, pc_{t_3,1}\}$ to a container $c_1$ of vendor \textsc{N} (not shown; cf. \cref{sub:motivating_example} for detailed evaluation and discussion) and $\{pc_{t_1,2}, pc_{t_1,3}\}$ to container $c_1$ of vendor $1$.
	\begin{figure}
\begin{minipage}[tb]{0.35\columnwidth}
		\small
		\captionof{table}{Processes (per tenant)}
		\label{tab:example_capacity}
		\begin{tabular}{p{0.7cm} p{0.5cm} p{0.5cm}}
			\hline
			Tenant & CAP & SH \\ \hline
			\multirow{3}{*}{$t_1$} & 13 & \faThumbsOUp \\
			& 19 & \faThumbsDown \\ 
			& 50 & \faThumbsDown \\
			\hline
			$t_2$ & 11 & \faThumbsOUp  \\ 
            \hline
			$t_3$ & 21 & \faThumbsOUp  \\ 
		\end{tabular}	
	\end{minipage}
	\begin{minipage}[tb]{0.6\columnwidth}
		\centering
		\includegraphics[width=1\columnwidth]{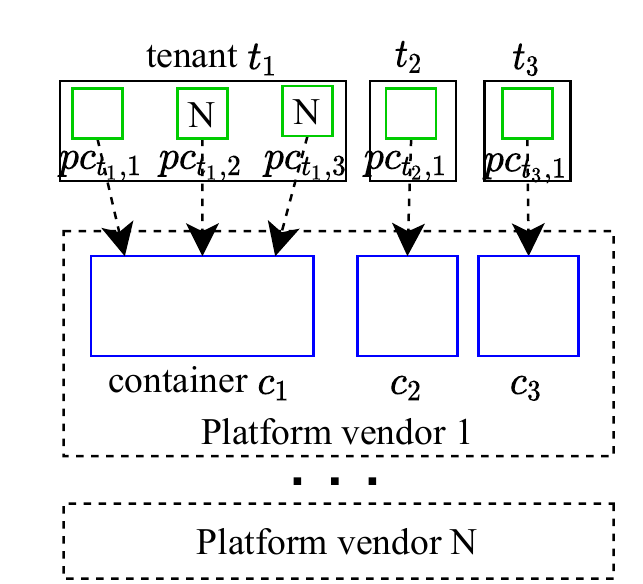}
		\captionof{figure}{Tenant-specific hosting}\label{fig:hosting}
	\end{minipage}
	\end{figure}
	\qedblack
\end{example}

We argue that for the problem of cost-efficient placement under constraints and customer preferences from the INTaaS perspective (cost-efficient process placement, short CEPP) \cite{DBLP:conf/edoc/000120}, multiclouds\footnote{Gartner predicted multicloud as common strategy for $70\%$ of enterprises by 2019, visited 2/2021: \url{https://gtnr.it/2K6yuhV}.} allow for more efficient solutions compared to the classical hosting or single-cloud approaches.
We claim that a multicloud placement of shareable sub-processes is profitable for the INTaaS vendor and can be reached without neglecting contractual obligations like security. 
At the same time, the INTaaS' customers seek for more insight into how their process modeling and configuration decisions impact their costs.
Having a complementary perspective, a user / process modeler has an interest in controlling the operational costs of the process models, besides the inherent complexity of modeling functionally correct process models.
For example, many countries, such as Italy, follow standardised procedures towards electronic invoicing (cf. \cref{ex:invoicing}) \cite{sap-hci-content}.
In particular, in Italy, the administrations accept electronic Business to Business and Business to Government invoices via FatturaPA\footnote{Fatturazione Elettronica verso la Pubblica Amministrazione (FatturaPA), visited 2/2021: \url{https://www.fatturapa.gov.it/it/index.html}}, which have to be sent through an interchange system called \emph{Sistema di Interscambio} (SdI) that records the data transfer.

\begin{example}
\label{ex:invoicing}
\begin{figure*}[bt]
	\centering
	\includegraphics[width=0.8\linewidth]{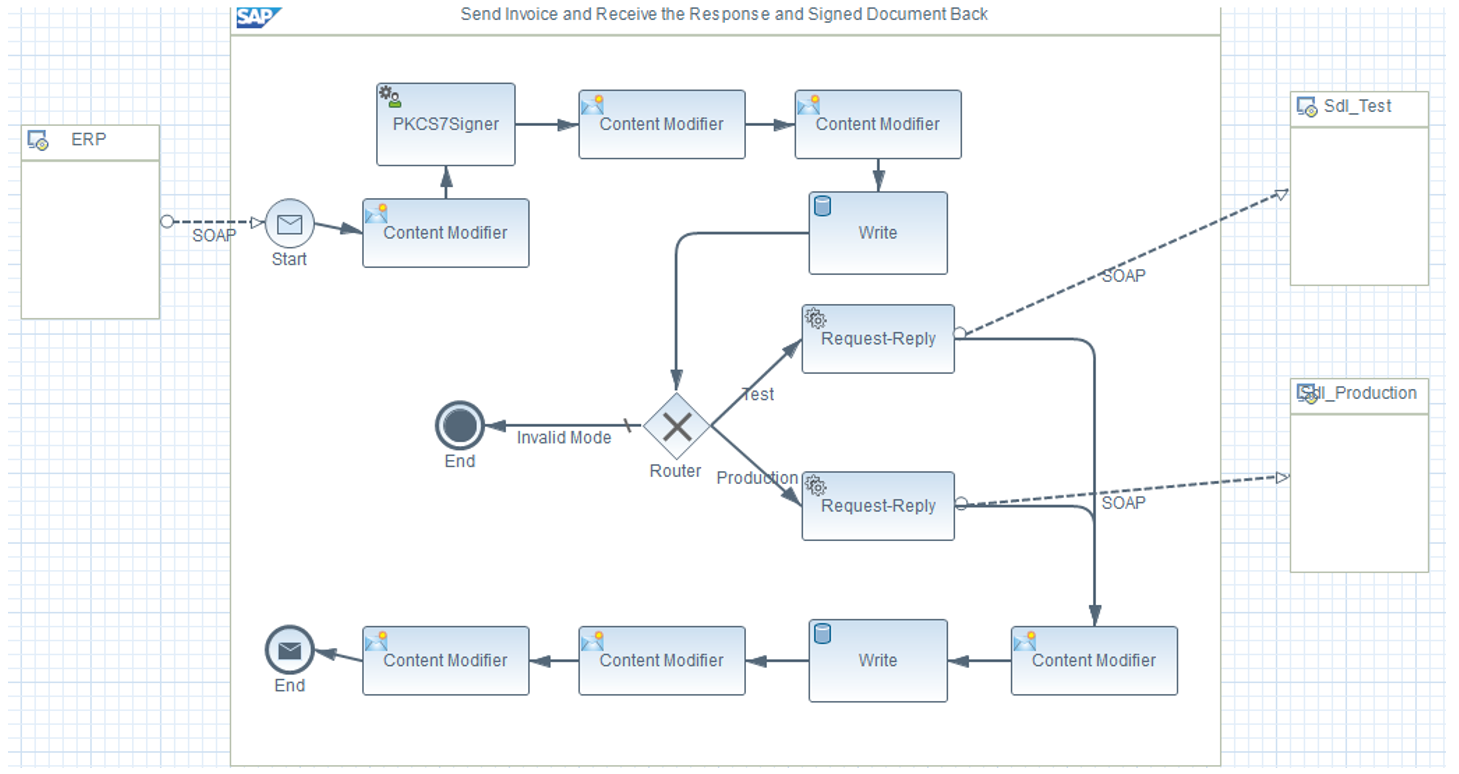}
	\caption{SAP CPI eDocuments Italy invoicing integration process (simplified from \cite{sap-hci-content})}
	\label{fig:cdd}
\end{figure*}
\Cref{fig:cdd} shows a simplified version of an invoicing integration process for Italy.
An ERP system sends invoices to the integration process, which \textbf{{\colorbox{amber!70}1}} preserves information like the transfer \textsl{mode}, as well as relevant identifiers like \textsl{IdentityCode} and \textsl{IdSender}, and then signs the invoice itself. 
In \textbf{{\colorbox{amber!70}2}} an SdI compliant message is prepared (\eg field \textsl{FiscalCode\_Client} mapped to \textsl{FiscalCode}) and the signed invoice is cached, before storing it for further reference. 
Message headers required by SdI are set (\eg \textsl{identityCode}) and \textbf{{\colorbox{amber!70}3}} either sent to a test or production endpoint, else discarded, depending on its mode.
The response message is enriched with information preserved from the original message (\eg signed invoice) \textbf{{\colorbox{amber!70}4}}, stored for reference and then \textbf{{\colorbox{amber!70}5}} a compliant response message for the ERP system is prepared (\eg \textsl{SignedDocument}, \textsl{DataTimeReception}). \qedblack 
\end{example}

While from an INTaaS customer's perspective, leveraging reduced costs of a multicloud offering could be beneficial, modeling integration processes such as electronic invoicing poses two important challenges to the user.
First, it is essential to model and configure the corresponding process functionally correct, \eg by generating valid invoices and avoiding penalties.
Second, the INTaaS' cost considerations and potentially their solution of the CEPP are intransparent to the user.
Consequently, the user is not aware of functionally correct, but potentially more cost-efficient process model configurations, and thus cannot answer important questions like \enquote{Is the process model cost-efficient?}, \enquote{What is the cost contribution of a dedicated process model?}, and \enquote{How to make the process more cost-efficient?}.

To overcome the challenges of the INTaaS vendor regarding costs and customer preferences like security, as well as the limited transparency of the INTaaS customer into the cost reduction potential of their process models, the solution of \emph{cost-efficient placement} has to be combined with the possibility of the customer and INTaaS vendor to interact, called \emph{cost-aware modeling}.
Therefore we adopt a \emph{cost-aware modeling process} that covers the following \emph{objectives}:
\begin{enumerate*}[label=(\roman*)]
    \item automatic placement decision,
    \item functional correctness-preserving process decomposition,
	\item considering multicloud platform vendors,
	\item security constraints,
	\item applicability to larger problem sizes (\ie performance),
	\item and cost improving process change proposals.
\end{enumerate*}
%
\begin{figure}[bt]
	\centering
	\includegraphics[width=\columnwidth]{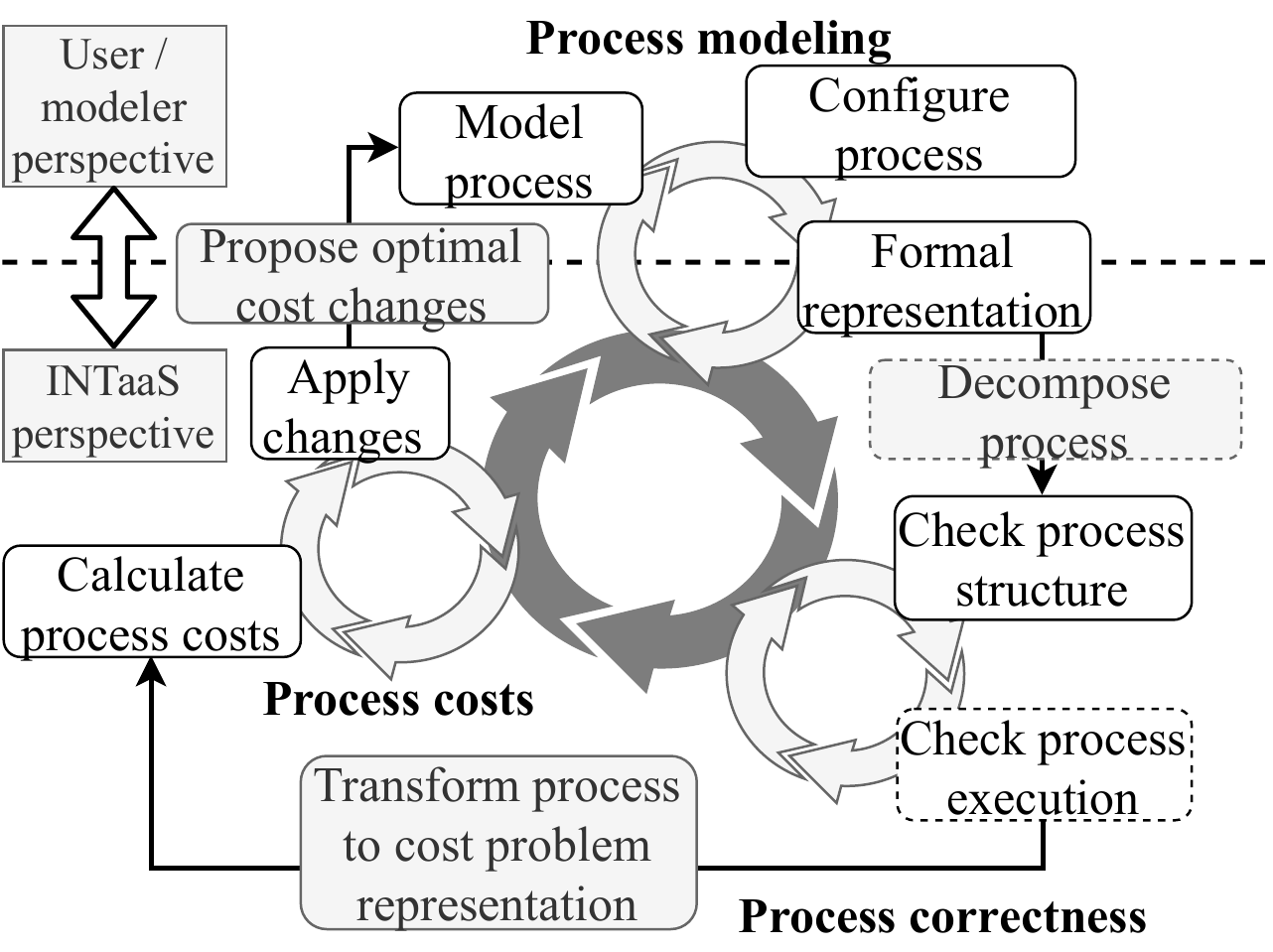}
	\caption{Cost-aware modeling process (white colored boxes denote the main steps, grey colored boxes represent explanatory steps or bridges between the main stages as well as the INTaaS and User / modeler perspectives, and dotted line boxes specify optional steps)}
	\label{fig:cost_aware_modeling_process}
\end{figure}
\Cref{fig:cost_aware_modeling_process} shows this process with its three main stages that we subsequently discuss: process modeling, process correctness, process costs (each stage is composed of process steps or actions).
\begin{itemize}
	\item[\emph{Modeling}] The integration process is modeled and configured in a process modeling tool, from which the model is then translated into a formal representation.
	With a thorough understanding of the security properties of a process, the formalized process model could be decomposed into smaller, but semantically equivalent process models for a more cost-efficient placement.
	\item[\emph{Correctness}] 
	The resulting (decomposed) formal process model is structurally and semantically analyzed and verified with respect to correctness (\ie simulation, model checking).
	For structural correctness, the formal representation together with the knowledge about the characteristics of composed patterns are sufficient.
	For the semantic correctness, a process is analyzed in the context of its dedicated, scenario-specific configuration, a test design, which specifies the desired properties with a notion of bi-simulation or assessment techniques like the expected output of a process simulation, for a given input.
	Any flaws found during this step result in another round of modeling or configuration adjustment.
	\item[\emph{Costs}] The correct integration process model is transformed into a cost problem representation, from which its costs for an INTaaS can be calculated optimally or through application of dedicated heuristics in an existing multicloud setup.
	This requires additional information (\eg about all currently operated processes by the INTaaS for all their customers) as well as their required resources and the PaaS vendors' container costs and capacities.
	These costs can be broken down either to an INTaaS customer or its end user, to find a more cost-efficient process model and
	lower cost process models can be proposed to the process modeler / integration developer as process improvements.
\end{itemize}

We argue that existing approaches do not fully support a cost-efficient placement, while allowing for a cost-aware modeling and hence the following research questions are formulated to guide the design and development of a cost-aware modeling process living up to objectives (i)--(vi):
\begin{itemize}
    \itemsep0em
    \item[Q1] When is a multicloud deployment beneficial for an INTaaS vendor? (cf. cost drivers; security constraints: (iii)--(iv)) 
    \item[Q2] How to practically determine a cost-efficient process placement in multiclouds? (cf. automatic; correctness-preserving process decomposition; large problem domains: (i)--(ii), (v))
    \item[Q3] How can processes be modeled cost-efficiently by a user for multiclouds? (cf. cost-aware modeling; vendor, user interaction (vi))
\end{itemize}
%
%
%
The conference paper \cite{DBLP:conf/edoc/000120} extended in this work has provided the foundations for Q1 (and partially Q2).
An optimal solution of CEPP has been formalized as mixed-integer linear program ($\rightarrow$ Q1) and a local search heuristic based on hill-climbing has been specified to make the solution practically applicable for real-world problem sizes ($\rightarrow$ Q2).
The evaluation showed that a multicloud placement and a decomposition of integration processes into smaller sub-processes are mostly beneficial for an INTaaS vendor in terms of cost-savings ($\rightarrow$ Q2).
However, decompositions were not fully specified, leaving the objective of \enquote{functional correctness} (cf. (ii)) open.

In this work, we add a correctness-preserving formal definition of process decompositions (\ie process remains functionally correct after decomposition) and cost-efficient process improvements ($\rightarrow$ Q2, Q3).
The specification of correctness-preserving integration process decompositions and process improvements is based on prior work on \emph{structural correctness} based on integration pattern contract graphs (IPCGs) \cite{DBLP:conf/debs/0001MFR18}.
However, a combined study with \emph{semantic correctness} \cite{DBLP:conf/edoc/0001RMRS18,ritter2019formal,DBLP:phd/at/Ritter2019} is out of scope and left for future work. 
For the first time, a combination of a practically tractable solution of CEPP (incl. process decomposition) and cost-reducing process improvements, enables the study of \emph{cost-aware process modeling} ($\rightarrow$ Q3).

The contributions to research questions Qx are elaborated following the principles of the design science research \emph{methodology} described in \cite{peffers2007design}: \enquote{\emph{Activity 1: Problem identification and motivation}} is based on literature and assessment of vendor driven solutions (\eg \cite{Ritter201736}) as well as quantitative analyses of security aspects of EAI building blocks (\ie existing catalogs of 166 integration patterns) and process improvements (\eg from \cite{DBLP:conf/debs/0001MFR18}).
\enquote{\emph{Activity 2: Define the objectives for a solution}} is addressed by formulating objectives (i)--(vi).
For \enquote{\emph{Activity 3: Design and development}} several artifacts are created to answer questions Q1--Q3, and meet objectives (i)--(vi), including (a) the specification of an extensible modeling framework for CEPP that can include new constraints, and comprises a benchmark for heuristic solutions, (b) a local search heuristic with a novel adaptation of bin-backing to CEPP for finding an initial solution, and four CEPP transformations: \emph{move}, \emph{shrink}, \emph{reduce}, and \emph{swap}, making the evaluation tractable and showing results on real-world data sets from SAP CPI \cite{sap-hci-content}, (c) a formal approach to \emph{correctness-preserving} process decomposition and formalization of two new process improvements (\ie \emph{combine patterns}, \emph{router to routing slip}), and (d) the definition and realization of a prototype, which we experimentally evaluate, and assess the potential of cost-aware process modeling in a case study.

The paper is structured as follows.
In \cref{sec:preliminaries} we analyze \emph{shareability} as operational security requirement of integration patterns and accordingly extend an existing process formalism from~\cite{DBLP:conf/debs/0001MFR18}, rate process improvements according to their cost-reduction potential, and state the CEPP problem for multiclouds for these processes. 
We formally define decomposition and improvement change primitives and discuss their correctness in \cref{sec:decomposition}.
We state CEPP as an optimization problem in \cref{sec:milp}, and specify a heuristic tailored to CEPP in \cref{sec:heuristic}.
In \cref{sec:evaluation}, we define and realize a prototype, which we experimentally evaluate on real-world processes and use to study cost-aware process modeling.
We discuss related work in \cref{sec:relatedwork} and conclude in \cref{sec:conclusion}.

\section{Shareable processes, cost improvements and multicloud model} \label{sec:preliminaries}
In this section we analyze integration patterns with respect to security constraint of side-effects (cf. objective (iv)) as well as specify the cost reduction potential of common optimization strategies as part of the cost-aware modeling process (cf. objective (vi)).
We briefly introduce the formal model for processes from \cite{DBLP:conf/debs/0001MFR18}, make the required extensions, and sketch the used multicloud model.

\subsection{Shareability and cost improvement analyses}

We conduct a pattern and process analysis regarding the security requirement of shareability, before we briefly summarize process cost improvements with a focus on two improvements required for our case study in \cref{sec:evaluation}.

\subsubsection{Shareability}
\label{sub:analysis}
The main security requirement that we target in this work is \emph{shareability}~\cite{alzain2012cloud} (\ie shareable or non-shareable patterns and processes), which denotes side-effects to another tenant.
We define integration patterns as free of side-effects from a security point of view (\ie shareable), if they do not allow for user defined functions (UDFs), other than the conditions and expressions offered by the INTaaS vendors, and do not access shared resources (\eg database).
The UDFs are critical in this regard, if they allow for programs that have side-effects on a shared container.
We assume that a process is \emph{shareable}, if all its patterns are free of side-effects.
Since shareability highly depends on the specific implementation of the pattern by the INTaaS vendor (\eg side-effect free transformation language vs. UDF), we base our analysis on the SAP CPI process content \cite{sap-hci-content}.

When classifying the $166$ integration patterns from the catalogs~\cite{Ritter201736,hohpe2004enterprise,DBLP:journals/ijcis/RitterS16} by a manual analysis of the pattern descriptions, there are $31$ patterns that are potentially non-shareable (\eg script, resequencer).
The specific implementations of the SAP CPI patterns reduces these to two (\eg script and message mapping), which are used in 48.59\% and 63.32\% of the processes for scripts and mappings, respectively.
The small number of patterns with side-effects is not surprising, since INTaaS vendors try to avoid them, but the high usage rate in CPI emphasizes their relevance for this provider, and thus making CEPP solutions desirable for CPI. 

\subsubsection{Cost-aware modeling improvements}
For the cost-aware modeling process in \cref{fig:cost_aware_modeling_process}, we consider automated changes and improvements as crucial.
After having calculated the process costs, the steps \enquote{Apply changes} and \enquote{Propose optimal cost changes} are relevant for the interaction with the modeler.
If the latter accepts the change proposals showing their cost reduction potentials, she can continue with the changed process model, conduct custom changes and get new change proposals, or stick to the current process model.
\begin{table}[tb]
	\centering
	\small
	\caption{Optimization strategies from \cite{DBLP:conf/debs/0001MFR18} in context of their cost reduction potential}
	\label{tab:optimization_strategies}
	\begin{tabular}{ll|c}
		\hline
          Strategy & Optimization & \parbox[t]{1.5cm}{Cost reduction} \\
		\hline
		\multirow{2}{*}{\parbox[t]{2cm}{OS-1: Process \\Simplification}} & \parbox[t]{2.5cm}{Redundant Sub-process Removal} & \faThumbsUp\\
			& \parbox[t]{3cm}{Combine Sibling Patterns} & \faThumbsUp \\
			& \parbox[t]{3cm}{Unnecessary conditional fork} & \faThumbsUp\\
			& \parbox[t]{3cm}{\textbf{Combine neighbors \cite{DBLP:journals/is/BohmHPLW11,habib2013adapting}}} & \faThumbsUp\\
			& \parbox[t]{3cm}{\textbf{Routing to routing slip (\cite{DBLP:journals/is/BohmHPLW11,vrhovnik2007approach})}} & \faThumbsUp\\
		\hline
		\parbox[t]{2cm}{OS-2: Data Reduction} & \parbox[t]{3cm}{Early-Filter} & \faThumbsUp\\
		& \parbox[t]{3cm}{Early-Mapping} & \faThumbsOUp\\
		& \parbox[t]{3cm}{Early-Aggregation} & \faThumbsOUp\\
		& \parbox[t]{3cm}{Claim Check} & \faThumbsUp\\
		& \parbox[t]{3cm}{Early-Split} & \faThumbsODown\\
		\hline
		\parbox[t]{2cm}{OS-3: Parallelization} &
		\parbox[t]{3cm}{Sequence to parallel} & \faThumbsODown\\
		& \parbox[t]{3cm}{Merge parallel sub-processes} & (\faThumbsUp)\\
	\end{tabular}
\begin{tablenotes}
	\centering
	\small
	\item \faThumbsUp: improvement, \faThumbsOUp: slight improvement, \faThumbsODown: deterioration, \textvisiblespace: no effect
\end{tablenotes}
\end{table}

In our previous work \cite{DBLP:conf/debs/0001MFR18}, we systematically collected and characterised optimizations applicable to integration processes to which we refer for a detailed introduction to optimization strategies OS-1 to OS-3 and their optimizations.
In this work, we set the optimizations into context of their cost reduction potential in \cref{tab:optimization_strategies} and add two optimizations required for a case study in our evaluation (\ie \emph{Combine neighbors}, \emph{Routing to routing slip}), which we adapted from the literature of related domains, and briefly describe subsequently.

The combination of neighbor patterns can also be found in the data integration (\eg \cite{DBLP:journals/is/BohmHPLW11}) and scientific workflow domains (\eg \cite{habib2013adapting}) as \emph{operator merge} (discussed in more detail in \cref{sub:cost_rewriting_rules}).
However, in our case the operators are essentially the integration patterns, and a merge is only allowed, if the neighbour patterns are of the same type (\eg subsequent message translator or content enricher patterns from \cite{hohpe2004enterprise}) and have no conflicts accessing the message.
Moreover, the simplification of a content based router to a routing slip pattern (both from \cite{hohpe2004enterprise}) is remotely related to \emph{Unnecessary Switch-Path} \cite{DBLP:journals/is/BohmHPLW11} and \emph{Eliminate-Unused-Partner} \cite{vrhovnik2007approach}  (discussed in more detail in \cref{sub:cost_rewriting_rules}).
While the intent seems similar, the routing slip maintains routing configurations for \emph{all} switch-paths of directly connected endpoints, and thus transforms process model elements into configurations.

Both new optimizations can be classified as OS-1: \emph{Process Simplification}, which group together all techniques whose main goal is reducing model complexity (\eg number of patterns~\cite{sanchez2010prediction}).
The cost reduction of these techniques can be measured by the reduced number of pattern instances in a runtime system that leads to a lower resource consumption (\eg main memory, CPU utilization), which is a major cost-driver in current cloud platforms (\eg \cite{buyya2010intercloud}).

Another way of reducing the amount of required memory is OS-2: \emph{Data Reduction}.
If data are either filtered before it arrives at the process, or data are filtered early in the process (\emph{Early filter}), or data are stored on disk and picked up at the end of the process (\emph{Claim check}), then the process could be placed on a smaller, cheaper container variant\footnote{A container variant denotes a platform vendor-specific combination of resources (\eg number of CPU cores, main memory in MB that is available for a certain cost. In contrast, a container is a specific processing unit of a certain variant.}.
Notably, the effects of the other OS-2 optimizations are smaller (\eg \emph{Early-Mapping}) or even require more memory (\eg \emph{Early-Split} consumes more memory for newly created messages while not reducing the overall amount of payload data).

Finally, OS-3: \enquote{Parallelization} increases the costs through redundancy (\eg \emph{Sequence to parallel}), and optimizations like \emph{Merge parallel sub-processes} merely invert cost inefficient parallelizations.

In summary, especially process simplification and data reduction strategies (\ie OS-1, OS-2) seem to have a cost reduction potential.
The newly introduced \emph{combine neighbors} and \emph{routing to routing slip} optimizations are defined in \cref{sec:decomposition}.

\vspace{-.2cm} \subsection{Integration process model} 
\label{sub:compositions}

We briefly recall integration pattern contract graphs (IPCGs)~\cite{DBLP:conf/debs/0001MFR18}, denoting a formal integration process model, and extend them with a shareability requirement and resource constraints.
We decided to use IPCGs due to their comprehensive coverage of integration processes and inherent structural correctness verification properties \cite{DBLP:conf/debs/0001MFR18} addressing objective (ii).
While IPCGs can only be used to formally assess the structural correctness of processes as well as their changes like improvements / optimizations \cite{DBLP:conf/debs/0001MFR18}, we showed in previous work \cite{DBLP:phd/at/Ritter2019} that they can be lifted to \tdbnet \cite{DBLP:conf/edoc/0001RMRS18,ritter2019formal}, a formalism based on Coloured Petri Nets, allowing for a formal assessment of their semantic correctness (out of scope in this work).

\subsubsection{Integration pattern contract graphs}
\label{sub:ipcg}
An integration pattern typed graph (IPTG) \cite{DBLP:conf/debs/0001MFR18} (cf. \cref{def:ipg}) is a connected, directed acyclic graph, whose nodes denote integration patterns ($P$) and the edges stand for message channels ($E$).
%
\begin{definition}[IPTG~\cite{DBLP:conf/debs/0001MFR18}]
  \label{def:ipg}
  An \emph{integration pattern typed graph} (IPTG) \cite{DBLP:conf/debs/0001MFR18} is a directed graph  with set of nodes $P$ and set of edges $E \subseteq P\times P$, together with a function $\type : P \to T$, where $T = \{$start, end, message processor, fork, structural join, condition, merge, external call$\}$.
  The cardinality of incoming edges in a node $p \in P$ is denoted by $\bullet p$ and and $p \bullet$ represents the cardinality of the outgoing edges of $p$.
  The variable $n \in \mathbb{N}$ specifies an arbitrary value greater than one. 
  An IPTG $(P, E, \type)$ is \emph{correct} if
  \begin{compactitem}[$\bullet$]
    \itemsep0em
    \item $\exists$ $p_1$, $p_2$ $\in$ P with $\type(p_1)$ = start and $\type(p_2)$ = end;
    \item if $\type(p) \in$ \{fork, condition\} then $|\bullet p|=1$ and $|p \bullet| = n$, and if $\type(p) = join$ then $|\bullet p| = n$ and $|p \bullet| = 1$;
    \item if $\type(p) \in$ \{message processor, merge\} then $|\bullet p| = 1$ and $|p \bullet| = 1$;
    \item if $\type(p) \in$ \{external call\} then $|\bullet p| = 1$ and $|p \bullet| = 2$;
    \item The graph $(P, E)$ is connected and acyclic. \qedwhite
  \end{compactitem}
\end{definition}
The nodes are \emph{typed} by pattern types with characteristics (\eg channel and message cardinalities), as defined in \cref{def:characteristics}. 
\begin{definition}[Pattern characteristic and configuration~\cite{DBLP:conf/debs/0001MFR18}]
  \label{def:characteristics}
  A \emph{pattern characteristic} assignment for an
  IPTG $(P, E, \type)$ is a function
  $\pchar: P \to 2^{\PC}$, assigning to each pattern a subset of the set
  \begin{align*}
    \PC =\
    & (\{\text{MC}\}\times \mathbb{N}^2) \cup {} 
      (\{\text{ACC}\} \times \{\text{ro},\text{rw}\}) \cup {} \\
    &  (\{\text{MG}\} \times \Bool) \cup {} 
     (\{\text{CND}\} \times 2^{\BExp}) \cup {} 
     (\{\text{PRG}\} \times \Exp ) \enspace ,  
  \end{align*}
  where $\Bool$ is the set of Booleans, $\BExp$ the set of Boolean
  expressions, $\Exp$ the set of program expressions, and $\text{MC}$, $\text{CHG}$, $\text{MG}$, $\text{CND}$, $\text{PRG}$, where characteristic $(\text{MC}, n, k)$ represents a message cardinality of $n$:$k$ with $n, k \in \mathbb{N}$, $(\text{ACC}, x)$ the message access,
depending on if $x$ is read-only $\text{ro}$ or read-write $\text{rw}$, and the characteristic $(\text{MG}, y)$ represents whether the pattern is message generating depending on the Boolean $y$. 
A $(\text{CND}, \{c1, ..., cn\})$ represents the conditions $c_1$, \ldots, $c_n$ used by the pattern to route messages, and $(\text{PRG}, (p))$ represents a program $p$ used by the pattern. \qedwhite
\end{definition}
The (cross-) pattern data composition is represented by inbound and outbound pattern contracts (\ie $inContr$, $outContr$) as defined in \cref{def:pattern-contract}.
\begin{definition}[Pattern contract~\cite{DBLP:conf/debs/0001MFR18}]
  \label{def:pattern-contract}
  A \emph{pattern contract} assignment for an IPTG $(P, E, \type)$ is a function $\contr : P \to {\CPT} \times 2^{\EL}$, assigning to each pattern a function of type
  \[
    \CPT = {\{\text{signed}, \text{encrypted}, \text{encoded}\} \to \{\text{yes}, \text{no}, \text{any}\}}
  \]
  and a subset of the set
  \[
    \EL = (\{\text{HDR}\} \times 2^{D}) \cup (\{\text{PL}\} \times 2^{D}) \cup (\{\text{ATTCH}\} \times 2^{D})
  \]
  where $D$ is a set of  data elements (the concrete elements of $D$ are not important, and will vary with the application domain).
  The set $\CPT$ in a contract represents integration concepts, while
  the set $\EL$ represents data elements and the structure of the
  message: its headers $(\{\text{HDR}\}, H)$, its payload $(\{\text{PL}\}, Y)$ and its attachments $(\{\text{ATTCH}\}, A)$, where $H$ is a key-value pair, $Y$ is an arbitrary payload value, and $A$ is a key-binary value pair.
\qedwhite
\end{definition}
%
Each pattern will have an inbound and an outbound pattern contract (later illustrated in \cref{ex:iptg}),
describing the format of the data it is able to receive and send
respectively --- the role of pattern contracts is to make sure that adjacent inbound and outbound contracts match (cf. \cref{def:matching-contracts}).
\begin{definition}[\cite{DBLP:conf/debs/0001MFR18}]
\label{def:matching-contracts}
Let $(C, E) \in {\CPT} \times 2^{\EL}$ be a pattern contract, and $X \subseteq {\CPT} \times 2^{\EL}$ a set of pattern contracts. Write $X_{\CPT} = \{ C'\ | \ (\exists E')\ (C', E') \in X \}$ and $X_{\EL} = \{ E'\ | \ (\exists C')\ (C', E') \in X \}$.
We say that $(C, E)$ \emph{matches} $X$, in symbols $\match((C, E), X)$, if the following condition holds:
\begin{align*}
&(\forall x)\big(C(x) \neq \text{any} \implies \\ & \qquad\quad(\forall C' \in X_{\CPT})(C'(x) = C(x) \vee C'(x) = \text{any}) \big) \land {} \\
	& (\forall (m, Z) \in E)\big( Z = \bigcup_{(m, Z') \in \cup X_{\EL}}Z' \big)
\end{align*}\qedwhite
\end{definition}
This results in the definition of an integration pattern contract graph (IPCG) (cf. \cref{def:ipcg}).
%
\begin{definition}
  \label{def:ipcg}
  An \emph{integration pattern contract graph} (IPCG) is a tuple
  \[
    (P, E, \type, \pchar, \iC, \oC)
  \]
  where $(P, E, \type)$ is an IPTG, $\pchar : P \to 2^{\PC}$ is a pattern characteristics assignment, and $\iC : \prod_{p \in P}({\CPT} \times 2^{\EL})^{|\bullet p|}$ and $\oC : \prod_{p \in P}({\CPT} \times 2^{\EL})^{|p \bullet|}$ are pattern contract assignments --- one for each incoming and outgoing edge of the patttern, respectively --- called the inbound and outbound contract assignment respectively.
  It is \emph{correct}, if the underlying IPTG $(P, E, \type)$ is correct, and  inbound contracts matches the outbound contracts of the patterns' predecessors, i.e.\ if for every $p \in P$
\[
  p = \text{start} \vee \match(\iC(p), \{ \oC(p')\ |\ p' \in \bullet p\}) \enspace .
\]

Two IPCGs are \emph{isomorphic} if there is a bijective function between their patterns that preserves edges, types, characteristics and contracts. \qedwhite
\end{definition}
Two IPCGs are \emph{isomorphic} if there is a bijective function between their patterns that preserves edges, types, characteristics and contracts.
The syntactic correctness of an IPCG follows from \cref{def:matching-contracts}, which defines a structural $match$ function over the contracts for correctness-preserving re-writing of the graphs (cf. \cref{ex:contracts}).
For more details on IPCG contracts beyond this paper we refer to \cite{DBLP:conf/debs/0001MFR18}.
\begin{example}
    \label{ex:contracts}
    \Cref{fig:sample_ipcg_contract} shows an IPCG with two contracts (neglecting the other contracts for better readability), the outbound contract \texttt{outContr} of the shareable content-based router pattern \tikz{\node[circle,draw,red,scale=0.5] (a) at (0,0) {\textcolor{white}{N}}} (cf. \texttt{CBR}) and the subsequent inbound contract \texttt{inContr} of the non-shareable multicast pattern \tikz{\node[circle,draw,red,scale=0.5] (a) at (0,0) {\textcolor{black}{N}}} (cf. \texttt{MC}).
    \begin{figure}[bt]
    	\centering
    	\includegraphics[width=\columnwidth]{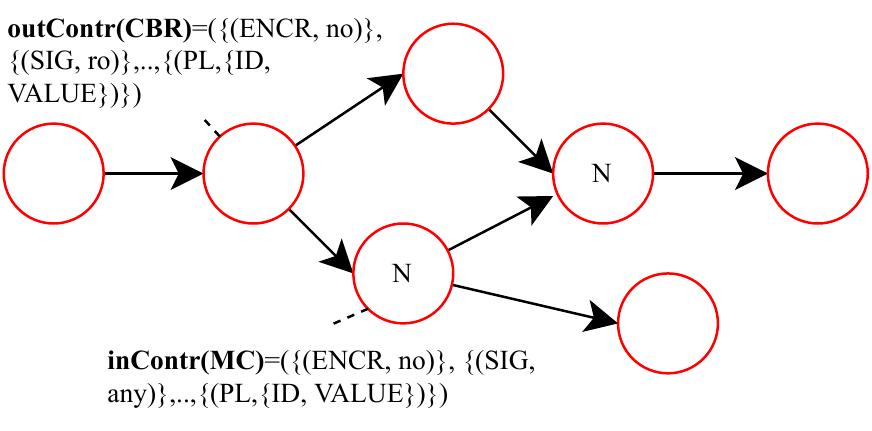}
    	\caption{Sample integration pattern contract graph (with two contracts)}
    	\label{fig:sample_ipcg_contract}
    \end{figure}
    The contracts depict properties indicating that messages passed from the router are neither encrypted (cf. \texttt{ENCR}) nor signed (cf. \texttt{SIG}) and have a payload \texttt{PL} of \texttt{ID} and \texttt{VALUE} data elements.
    The inbound contract of the multicast has a fitting inbound contract, which accepts only messages with unencrypted payload and the same data elements, and can handle signed or unsigned messages.
    In other words, the contracts match, and thus this part of the process is structurally correct.
    \qedblack
\end{example}

\labeltitle{Extensions (Shareability, Capacity)} The shareability requirement and resource constraints (as capacity) can be considered as additional pattern characteristics of the set of \textsl{pattern characteristics} of each pattern, formally defined as follows:
\begin{definition}[Extended pattern characteristic and configuration]
  \label{def:characteristics_extended}
  An extended \emph{pattern characteristic} assignment based on \cref{def:characteristics} for an
  IPTG $(P, E, \type)$ is a function
  $\pchar: P \to 2^{\PC}$, assigning to each pattern a subset of the set
  \begin{align*}
    \PC =\
    & (\{\text{MC}\}\times \mathbb{N}^2) \cup {} 
      (\{\text{ACC}\} \times \{\text{ro},\text{rw}\}) \cup {} \\
    &  (\{\text{MG}\} \times \Bool) \cup {} 
     (\{\text{CND}\} \times 2^{\BExp}) \cup {} 
     (\{\text{PRG}\} \times \Exp ) \cup {} \\
    & (\{\text{CAP}\}\times \mathbb{N}) \cup {} 
	(\{\text{SH}\} \times \Bool) {} \enspace ,
  \end{align*}
    where, in addition to \cref{def:characteristics}, capacity $\text{CAP}$ is a positive numeric value, indicating the amount of resources required by the pattern during runtime, and shareability indicator $\text{SH}$ denotes whether the pattern has security relevant side-effects. \qedwhite
\end{definition}
According to the analysis in~\cref{sub:analysis}, we add a Boolean characteristic that denotes whether the pattern is shareable or not (SH) and its required capacity (CAP) to the already existing characteristics (cf. \cref{ex:iptg}). 
While CAP is set automatically with runtime data (\eg used main memory, CPU), the valuation of SH is further discussed in~\cref{sec:evaluation}.
\begin{example}
\label{ex:iptg}
	\Cref{fig:sample_ipcg} denotes an integration pattern type graph (with characteristics, no contracts) that is essentially a composition of single patterns (\ie nodes) connected by message channels (\ie edges).
    \begin{figure}[bt]
    	\centering
    	\includegraphics[width=\columnwidth]{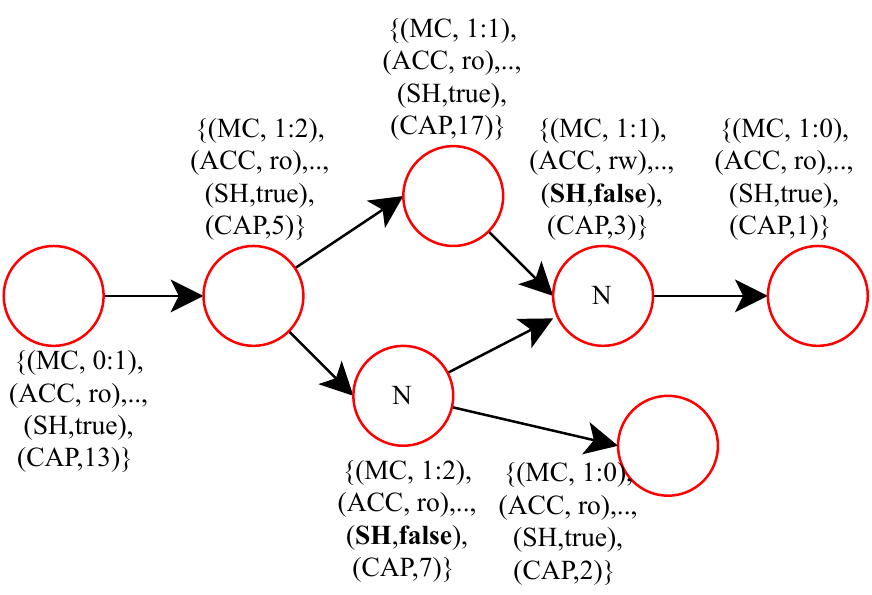}
    	\caption{Sample integration pattern type graph (with focus on the pattern characteristics)}
    	\label{fig:sample_ipcg}
    \end{figure}
	The nodes have characteristics like message cardinality \textsc{MC} and content change \textsc{CHG} among others as well as the new shareability \textsc{SH} and required capacity \textsc{CAP} characteristics.
	Briefly, the process starts with a pattern that creates a message (thus message cardinality of $0$:$1$) and forwards it to a read-only routing pattern with two outgoing message channels.
	Thereby, one of the successor patterns is shareable, while the other is assumed to have side-effects.
	Similarly, the join router implementation in this example has side effects, which is usually not the case \cite{Ritter201736}.
	The receiving patterns both have a $1$:$0$ cardinality. 
	The capacities denote resources required for these patterns during runtime (\eg main memory / RAM). \qedblack
\end{example}
The structural representation of an IPCG is close to the process models a user usually work with and subsequently serves as a basis to define our CPEE multicloud model as well as for decomposition and optimization definitions.

\subsection{Multicloud INTaaS model}
Similar to~\cite{buyya2010intercloud,singhal2013collaboration,alzain2012cloud} we consider the term \enquote{multicloud} to denote applications distributed over multiple cloud platforms.
Further --- according to the NIST specification~\cite{mell2011nist} --- we name vendors of these platforms \emph{platform-as-a-service providers}, which offer computing environments, called containers, with different capacities and prices as container costs that we specify as in \cref{def:multicloud}, allowing to identify the variants across vendors.
%
\begin{definition}[Multicloud cost configuration]
\label{def:multicloud}
    A tuple $(j, \omega_j, B_j, G_j)$ specifies a multicloud cost configuration with a set of computing environments $\omega_j$ of which $j \subseteq \omega_j$ denotes a specific container, a set of numberic container capacities $B_j$, and costs $G_j$. \qedwhite
\end{definition}
We consider processes based on IPCGs from~\cref{sub:compositions} (cf. \cref{def:process}), whose process model and content is developed by the customer (called tenant), \eg in an INTaaS vendor-specific integration language.
The process has a required capacity and a shareability property (\ie side-effects of process $k$ in the same container $j$ as process $i$) that is derived from the underlying IPCG.
\begin{definition}[Multicloud process]
\label{def:process}
    Tuple $(i,T_i,A_i,z_{ikj})$ represents a multicloud process, where $i$ denotes an IPCG (cf. \cref{def:ipcg}), a tenant assignment $T_i$, an assignment of the required capacity $A_i$ (\eg DRAM in MB) and the IPCG's shareability property $z_{ikj}$. \qedwhite
\end{definition}
The required capacity of a process $A_i$ denotes the sum of the capacities of all patterns in $i$.
We claim that a multicloud placement of processes is beneficial for an INTaaS vendor and argue that a decomposition (cf. \cref{ex:composition_partitioning}) could lead to an even more cost-efficient placement.
\begin{example}
\label{ex:composition_partitioning}
	In \cref{fig:composition_partitioning} the different conceptual entities are set into context by the example of the IPCG from~\cref{fig:sample_ipcg} for a given set of platforms from vendors $\omega$ and $\omega+1$.
    \begin{figure}[bt]
    	\centering
    	\includegraphics[width=\columnwidth]{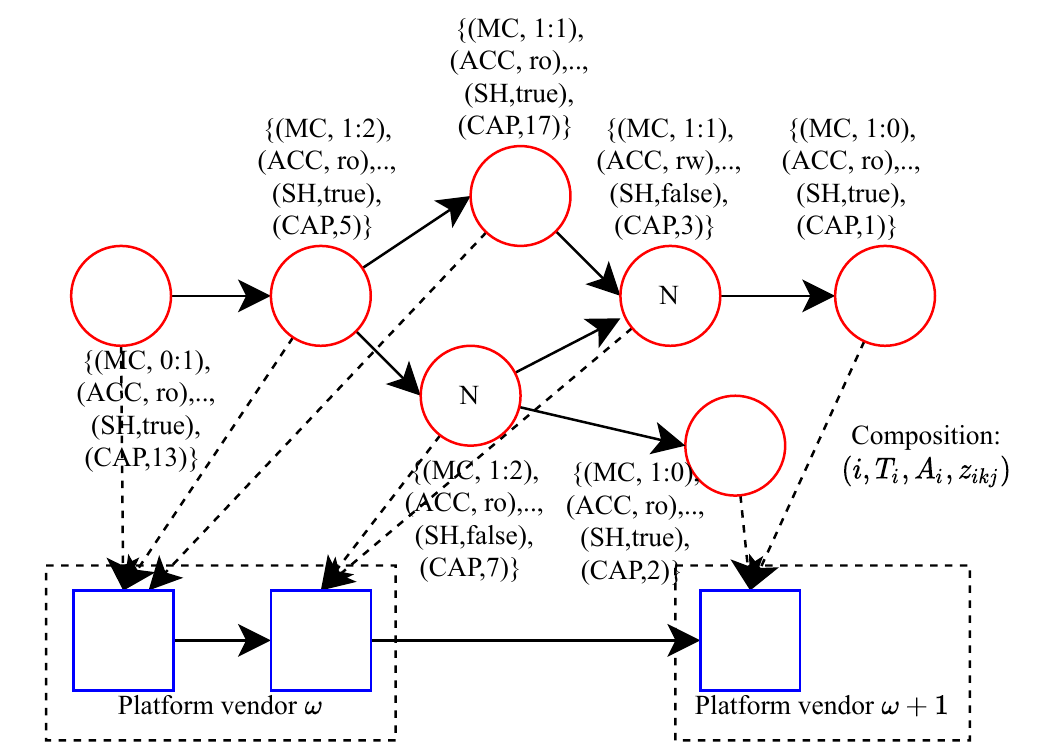}
    	\caption{Placement of the pattern (sub-) processes on platform variants}
    	\label{fig:composition_partitioning}
    \end{figure}
	The example assumes that the process can be decomposed across the two cloud vendors and that the decomposition makes sense from a cost perspective.
	Consequently both configurations of platform vendor take some of the pattern sub-processes (cf. dashed arrows).
	Vendor $\omega$ runs the non-shareable process parts isolated to fulfill the no information sharing constraint.
	The non-dashed lines denote message exchange, whose latency increases with a security-aware decomposition (\eg across multiple platform vendors). \qedblack
\end{example}
Generally, a decomposition of an IPCG is assumed to lead to multiple dependent IPCGs (calling each another), which is introduced and discussed in depth in \cref{sec:decomposition} and later evaluated and referred to as \emph{cutting} in \cref{sec:evaluation}.

\section{Process decomposition and improvements}
\label{sec:decomposition}
In this section we formally define the decomposition of processes and the two new improvements from \cref{sec:preliminaries} as changes of IPCGs based on \cite{DBLP:conf/debs/0001MFR18} and discuss their structural correctness-preserving properties (cf. objectives (ii), (vi)).

\subsection{Formal representation}
The decomposition or cutting of a problem into smaller problem sizes is a powerful model-based improvement or optimization technique (\eg \cite{korte2012combinatorial}).
Similar to \enquote{component decomposition} in Silva et al. \cite{DBLP:conf/europar/SilvaP17}, integration processes are decomposed into smaller sub-processes and then a solution for the CEPP of the sub-processes is calculated (cf. \cref{ex:composition_partitioning}).
This increases their number, but leads to smaller processes that might be placed more cost-efficiently.
However, the decomposition of a process must not impact the correctness of its combined sub-processes.
For that we first recall the background on changes to IPCGs using a rule-based graph rewriting system from \cite{DBLP:conf/debs/0001MFR18}, giving a formal framework in which different changes can be compared.
We begin by describing the graph rewriting framework, subsequently applying it to define decompositions and improvements.

\subsubsection{Background: Graph Rewriting}
Graph rewriting provides a visual framework for transforming graphs in a rule-based fashion.
A graph rewriting rule is given by two embeddings of graphs $L \hookleftarrow K \hookrightarrow R$, where $L$ represents the left hand side of the rewrite rule (called pattern graph), $R$ the right hand side (called replacement graph), and $K$ their intersection (the parts of the graph that should be preserved by the rule).
A rewrite rule can be applied to a graph $G$ after a match of $L$ in $G$ has been given as an embedding $L \hookrightarrow G$; this replaces the match of the pattern graph $L$ in $G$ by replacement graph $R$.
The application of a rule is potentially non-deterministic: several distinct matches can be possible~\cite{Ehrig:2006:FAG:1121741}.
Visually, we represent a rewrite rule by a left hand side and a right hand side graph colored green and red: green parts are shared and represent $K$, while the red parts are to be deleted in the left hand side, and inserted in the right hand side respectively.
For instance, the following rewrite rule moves the node $P_1$ past a fork by making a copy in each branch, changing its label from $c$ to $c'$ in the process in \cref{fig:example_rewriting}.
\begin{figure}[h!]
	\centering
	\includegraphics[scale=0.5]{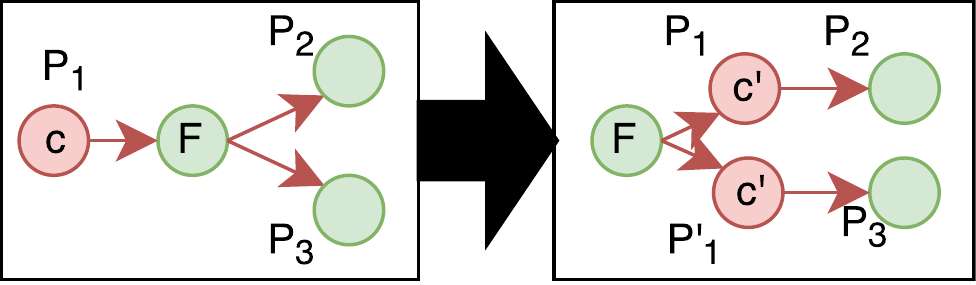}
	\caption{Process rewriting example}
	\label{fig:example_rewriting}
\end{figure}
%
%
%
Formally, the rewritten graph is constructed using a double-pushout (DPO)~\cite{ehrig1973graph} from category theory.
We use DPO rewriting since rule applications are side-effect free (\eg no \enquote{dangling} edges) and local (\ie all graph changes are described by the rules).
We additionally use Habel and Plump's relabeling DPO extension~\cite{habel2002relabelling} to facilitate the relabeling of nodes in partially labeled graphs.
\Cref{fig:sample_ipcg_contract,fig:sample_ipcg} show contracts and characteristics as comments, but in the rules that follow, we will represent them as (schematic) labels inside the nodes for space reasons.

In addition, we also consider rewrite rules parameterized by graphs, where we draw the parameter graph as a cloud (see \eg \cref{fig:nonshareable_rewriting_1} for an example).
A cloud represents any graph, sometimes with some side-conditions that are stated together with the rule.
When looking for a match in a given graph $G$, it is of course sufficient to instantiate clouds with subgraphs of $G$ --- this way, we can reduce the infinite number of rules that a parameterized rewrite rule represents to a finite number.
Parameterized rewrite rules can formally be represented using substitution of hypergraphs~\cite{plump1994hypergraph} or by !-boxes in open graphs~\cite{bangBoxes}.
Since we describe optimization strategies as graph rewrite rules, we can be flexible with when and in what order we apply the strategies.
We apply the rules repeatedly until a fixed point is reached, \ie when no further changes are possible, making the process idempotent.
Each rule application preserves IPCG correctness in the sense of \cref{def:ipcg}, because input contracts do not get more specific, and output contracts remain the same.
For general, formal foundations of correctness-preserving graph transformation we refer to, \eg \cite{DBLP:conf/calco/OrejasGLE09}.
Methodologically, the rules are specified by pre-conditions, change primitives, post-conditions and an optimization effect, where the pre- and post-conditions are implicit in the applicability and result of the rewriting rule.

\subsubsection{Graph rewriting rules for decomposition}
\label{sub:decomposition_rules}
For the decomposition of integration processes into shareable and non-shareable parts, we focus on the transitions from a shareable pattern to a connected subgraph of non-shareable patterns, and vice versa.

\paragraph{Shareable to non-shareable transition}
The transition from a shareable pattern to a subgraph of non-shareable patterns adds a remote call to the shareable part and a receiving endpoint to the non-shareable part.
For the transformation to be correct, the original IPCG is split into two IPCGs, connected by a remote call configuration. 
\labeltitle{Change primitives:}
The rewriting is depicted by the rule in \cref{fig:nonshareable_rewriting_1}\footnote{We briefly repeat the notation for a better understanding: nodes denote patterns ($P_1$, .., $P_n$) and node clouds stand for subgraphs of several nodes ($N_1$). Edges represent message channels.
The coloring depicts the interface or intersection of the left and right hand sides of the rule: red elements change, green elements remain unchanged.}.
The left hand side of the rule detects a shareable pattern $P_1$ that is connected to one non-shareable pattern that itself might be connected to several non-shareable patterns without a shareable one in-between.
\begin{figure}[h!]
	\centering
	\includegraphics[width=0.7\columnwidth]{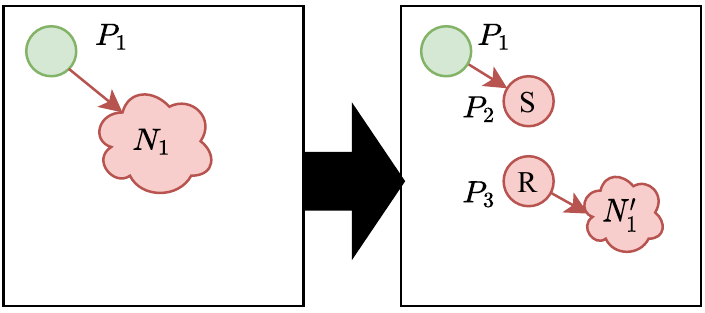}
	\caption{Shareable pattern to non-shareable rewriting}
	\label{fig:nonshareable_rewriting_1}
\end{figure}
In the right hand side of the rule, the direct call from $P_1$ to one pattern in $N'_1$ is replaced by an external call pattern $P_2$ \cite{Ritter201736} connected to $P_1$ that forwards the message to a newly added event-based receiver endpoint pattern $P_3$ \cite{hohpe2004enterprise} that forwards the message to the respective pattern in $N'_1$.

\labeltitle{Effect:} The change is especially beneficial for cases in which the size of the non-shareable subgraph $N'_1$ is greater than the sum of added external call and event-based receiver patterns, or if the resulting subgraphs have a more cost-optimal fit in the respective CPEE problem.

\paragraph{Non-shareable to Shareable transition}
Similarly, the transition from a subgraph of non-shareable patterns to a shareable pattern splits the original IPCG into two IPCGs connected via implicit remote communication configuration.

\labeltitle{Change primitives:}
In the rewriting rule depicted in \cref{fig:nonshareable_rewriting_2}, the left hand side matches any configuration with a transition from one non-shareable pattern (potentially connected to other non-shareable patterns) $N_1$ to a shareable pattern $P_1$.
\begin{figure}[h!]
	\centering
	\includegraphics[width=0.7\columnwidth]{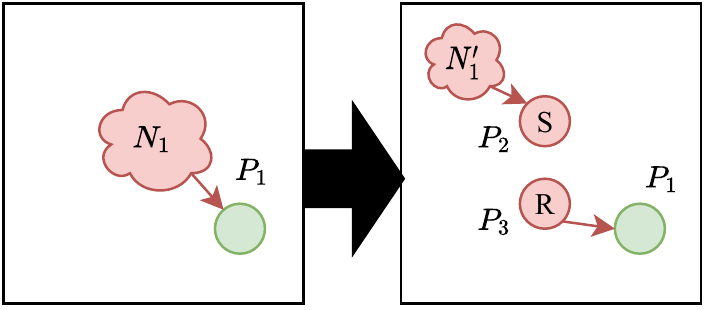}
	\caption{Non-shareable to shareable pattern rewriting}
	\label{fig:nonshareable_rewriting_2}
\end{figure}
The right hand side of the rule shows the replacement of the direct call between the patterns by an external call $P_2$ that connects the non-shareable subgraph $N'_1$ with the shareable $P_1$ via the receiver $P_3$.

\labeltitle{Effect:} same effect as for \emph{shareable to non-shareable transition}.

\Cref{ex:composition_partitioning_cut} shows the application of these rules to the process in \cref{fig:composition_partitioning}.
\begin{example}
\label{ex:composition_partitioning_cut}
When iteratively applying the two previously defined rules to the IPCG in \cref{fig:composition_partitioning}, the IPCG will be decomposed as shown in \cref{fig:composition_partitioning_cut}.
\begin{figure}[bt]
	\centering
	\includegraphics[width=\linewidth]{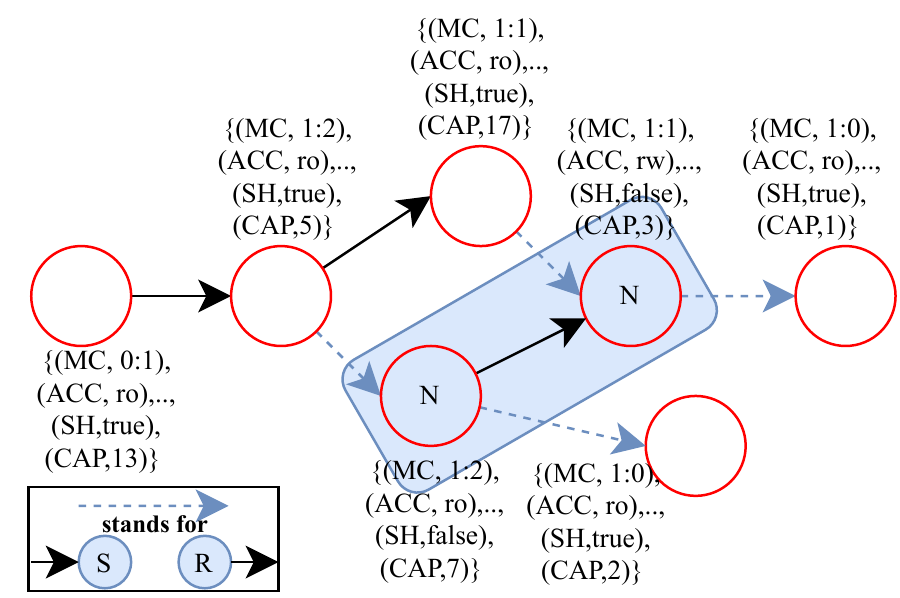}
	\caption{Decomposition of processes from \cref{ex:composition_partitioning}}
	\label{fig:composition_partitioning_cut}
\end{figure}
    The first rule matches the transition from the conditional fork to the first non-shareable pattern, by adding a remote communication (shown as dashed, blue arrow).
    Similarly, the second shareable to non-shareable transition will be disconnected, leaving a shareable IPCG with three of the original patterns and two external call patterns, connected to them.
    No further rules match on that IPCG.
    The resulting second IPCG with two connected non-shareable patterns has two matches of the non-shareable to shareable rule, eventually leading to two additional, shareable processes with one pattern each.
    In total there are three shareable and one non-shareable IPCG left for cost calculation.
\qedblack
\end{example}

\subsubsection{Graph rewriting rules for cost-aware modeling}
\label{sub:cost_rewriting_rules}
In \cref{tab:optimization_strategies} we listed known optimization strategies from our previous work \cite{DBLP:conf/debs/0001MFR18}, added two new optimizations and assigned their cost reduction potential.
Since the new optimizations (\ie \emph{combined neighbors}, \emph{reduce router to routing slip}) have not been formally defined in \cite{DBLP:conf/debs/0001MFR18}, we subsequently specify them as graph rewritings.
Consequently, they can then be used together with the already existing optimizations in \cite{DBLP:conf/debs/0001MFR18}.

\paragraph{Combine neighbors} Two patterns are considered neighbors, if they are sequentially connected to each other in an IPCG.
Neighbors can be combined into one pattern, if their configurations on an $\EL$ data element level are non-conflicting, non-dependent (\eg one pattern deletes a data element that another one tries to read) or their operations can be executed by the same pattern instance in the order of their appearance in the IPCG.
Hence, to combine two or more neighbor patterns, we require that they are all of the same pattern type (\eg content enricher).
We are currently aware of two message tranformation patterns in SAP CPI \cite{sap-hci-content} relevant to this work (\ie content enricher, message translator) that could be automatically combined.

\labeltitle{Change primitives:} The rule is given in \cref{fig:combine_patterns}, where sub-sequence $SSQ_1$ denotes a sequence of non-conflicting neighbor patterns of the same type that are combined to one pattern $P_3$ of that type.
\begin{figure}[h!]
	\centering
	\includegraphics[width=0.7\columnwidth]{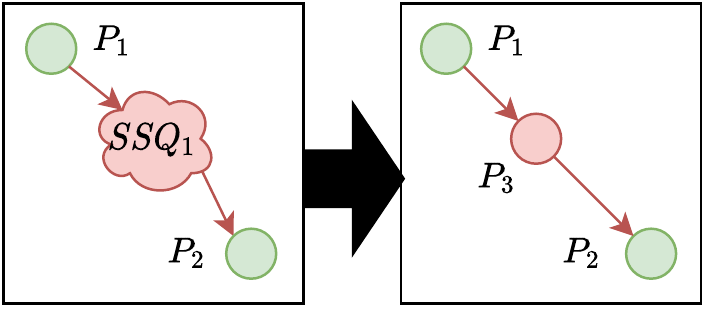}
	\caption{Combine pattern rewriting}
	\label{fig:combine_patterns}
\end{figure}

\labeltitle{Effect:} The model complexity is reduced by translating the explicitly modeled sequence of patterns to implicit configurations in just one pattern (\eg XPath).
Assuming that each pattern instance allocates redundant resources (\eg memory), the static resource consumption of this IPCG is reduced by the number of patterns in $SSQ_1$ minus the new pattern $P_3$.
Additional cost reduction potential lies in merging resource artifacts like files (\eg XSD, WSDL) to just one instance each.
When these resource artifacts are loaded into a process instance, they lead to higher main memory consumption, and thus eventually costs.

\paragraph{Reduce router to routing slip} If a content-based router is directly connected to one or many external call patterns having the same inbound contract (\ie usually different endpoints of one type of receiver system), then the router and all external call patterns can be replaced by a content enricher pattern, which configures the recipient according to the conditions of the router, followed by just one external call pattern that takes the configuration into account and forwards the message to the respective receiver.

\labeltitle{Change primitives:} The rule is given in \cref{fig:router_routing_slip}, where $P_1$ is a content-based router with $n-1$ conditional branches to external call patterns $\{P_2$, \dots, $P_n\}$ to the same receiver $R$.
\begin{figure}[h!]
	\centering
	\includegraphics[width=0.7\columnwidth]{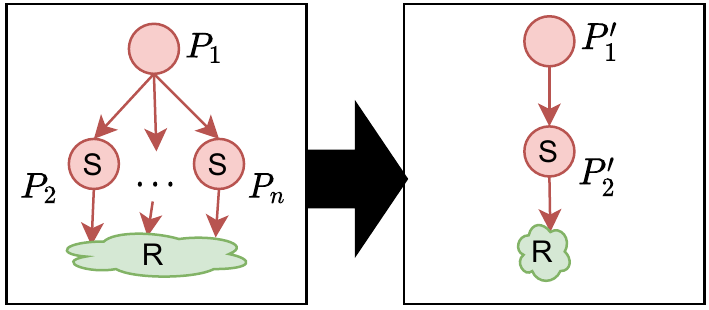}
	\caption{Router to routing slip pattern rewriting}
	\label{fig:router_routing_slip}
\end{figure}
On the right hand side of the rule, a content enricher $P'_1$ is added instead of the router $P_1$ and configured to add a routing slip to incoming messages according to the conditions from $P_1$.
Subsequently, only one external call pattern $P'_2$ is required, which checks the routing slips in the message to forward to the correct receiver in $R$.

\labeltitle{Effect:} Similar to the combine patterns rewriting, the model complexity is reduced, and thus shrinking the required pattern instances to two, instead of $n+1$ patterns in the explicit routing case.

\subsection{Discussion --- correctness}
\label{sub:correctness}

The correctness of the introduced process changes (\ie process remains functionally correct after changes have been applied), represented by graph rewriting rules, is crucial for the cost-aware modeling process (cf. \cref{fig:cost_aware_modeling_process}).
With the representation of integration processes as IPCGs we ensured that they are structurally correct by definition (cf. \cref{def:ipcg}).
However, we did not yet show that IPCGs remain structurally correct after applying the introduced rewritings rules.
Hence, we subsequently discuss the correctness properties of each graph rewriting individually based on the definition of matching contracts in IPCGs.
\begin{lemma}
    \label{lemma:rewritings}
	The rewritings in~\cref{sub:decomposition_rules,sub:cost_rewriting_rules} result in correct processes, if their pattern contracts match according to \cref{def:ipcg} and if all newly created graphs remain reachable.
\end{lemma}
\begin{proof}[Proof (sketch)]
    In each step, the rewriting modifies certain nodes.
    By \cref{lemma:rewritings}, it is enough to show that the affected parts of the interpretation of the graphs retain matching pattern contracts and all decomposed parts remain reachable as in the original graph, which we show for each rewriting subsequently.
    
    \labelsubtitle{Shareable to non-shareable transition} The right hand side of the rule in \cref{fig:nonshareable_rewriting_1} introduces a remote call pattern $S$ connected to $P_1$ and an event-based receiver pattern $R$ connected to a pattern in $N'_1$.
    The process is correct, if the inbound pattern contract $inContr(S)$ of the introduced external call pattern $S$ accepts any $outContr(P_1)$ of $P_1$:
    \begin{align*}
        inContr(S) = & ({(ENCR, any)}, {(SIG, any)},..,\\
                     & {(PL,any)}) \enspace , \\
        outContr(S) = & ({(ENCR, any)}, {(SIG, any)},..,\\
                     & {(PL,any)}) \enspace ,
    \end{align*}
    and its outbound contract $outContr(S)$ passes any input onward without introducing further restrictions in the contract.
    Further, the inbound pattern contract $inContr(R)$ of the receiver pattern $R$ accepts any $outContr(S)$ of $P_2$ 
    \begin{align*}
        inContr(R) = & ({(ENCR, any)}, {(SIG, any)},..,\\
                     & {(PL,any)}) \enspace , \\
        outContr(R) = & ({(ENCR, any)}, {(SIG, any)},..,\\
                     & {(PL, any)}) \enspace ,
    \end{align*}
    and the event-based receiver's outbound contract $outContr(R)$ passes any input onward without introducing further restrictions in the contract for a given inbound contract $inContr(N_1)$ with
    \begin{align*}
        inContr(N'_1) = & ({(ENCR, X)}, {(SIG, Y)},..,\\
                     & {(PL,Z)}) \enspace , \\
        outContr(P_1) = & ({(ENCR, X)}, {(SIG, Y)},..,\\
                     & {(PL,Z)}) \enspace ,
    \end{align*}
    as well as $outContr(P_1)$ matches $inContr(N'_1)$.
    Finally the remote call configuration of the external call point to the event based receiver, making the remote IPCG reachable.
    
    \labelsubtitle{Non-shareable to shareable transition} Same as for \emph{shareable to non-shareable} transition, but in reverse direction, starting from one pattern in $N_1$.
    
    \labelsubtitle{Combine patterns} The rewrite rule in \cref{fig:combine_patterns} is correct, if the pattern contracts of the patterns in the sub-sequence $SSQ_1$ can be combined to an inbound contract $inContr(P_3)$ of pattern $P_3$ that matches $outContr(P_1)$ of the preceding pattern $P_1$ and an outbound contract $outContr(P_3)$ matching the $inContr(P_2)$ of the succeeding pattern $P_2$.
    That is given through the rule matching only patterns of the same type and with non-conflicting, non-dependent changes to the message elements.
    That means, the pattern dependent contract parts like $\{(ENCR, yes)\}, \{(SIG, any)\}$ are all the same due to the same type and each pattern either modifies different data elements in $\{(PL,Z)\}$ (\eg a pattern does not delete data elements required by a subsequent) or requires read-only access, and thus the union of all accessed data elements denotes the inbound contract and the outbound contract is just the one of the last pattern in the sequence, satisfying the successor's inbound contract.
    
    \labelsubtitle{Router to routing slip} The rewrite rule in \cref{fig:router_routing_slip} is correct, if the outbound contract $outContr(P'_1)$ of the newly added content enricher $P'_1$ matches the inbound contract of the receiver $inContr(R)$, and the pattern contracts of the external call pattern $S$ accepts any configuration, same as for the transition rules before.
    To replace the router on the left hand side of the rule $P_1$ with an outbound contract of 
    \begin{align*}
        outContr(P_1) = & ({(ENCR, X)}, {(SIG, Y)},..,\\
                     & {(PL, Z)}) \enspace ,
    \end{align*}
    the enricher on the right hand side $P'_1$ requires the same outbound contract, but will add an additional header \texttt{HDR} with the routing information that was explictly modeled before:
    \begin{align*}
        outContr(P'_1) = & ({(ENCR, X)}, {(SIG, Y)},..,\\
                     & {(HDR, URL),(PL, Z)}) \enspace .
    \end{align*}
    Its subsequent external call pattern $S$ requires the header to direct the message to the correct endpoint, but either remove the header to have the same $outContr(P_1)$ as $P_1$ on the left hand side of the rule or the receiver has to ignore the additional header by having $\{(HDR, any),(PL, Z)\}$ in its inbound contract.
    That is given through the rule matching only endpoints in $R$ with the same pattern contract (but different address).
\end{proof}

While structural correctness is an important property built-in to IPGCs, we reiterate that semantic correctness would be meaningful, but out of scope in this work, and thus refer to \cite{DBLP:conf/edoc/0001RMRS18,ritter2019formal,DBLP:phd/at/Ritter2019}.

\section{Cost-efficient process placement}
\label{sec:milp}

In this section we specify the CEPP problem and propose a cost model for an optimal solution of the CEPP, which we base on the multicloud INTaaS model (cf. objectives (i), (iii)--(iv)).

\subsection{Placement problem}
For Software-as-a-Service vendors (incl. INTaaS), the CEPP is a bin-packing problem, with a set of bins of different sizes and costs (\eg platforms or variants), and a set of items of different sizes (\eg pattern (sub-) processes). 
The objective is to find a cost efficient placement or assignment of the processes under the following constraints:
\begin{inparaenum}[\it (C1)] 
	\item each process is in exactly one container,\label{item:flowConstraint}
	\item each container is exactly one variant,\label{item:containerConstraint}
	\item the required process capacity must not exceed the container size,\label{item:sizeConstraint}
	\item no non-shareable process is assigned to a container which contains processes of a different tenant (security-awareness),\label{item:securityConstraint}
	\item the number of processes per container must not exceed the container's maximum capacity, \label{item:maxFlowNumber}
\end{inparaenum}
which essentially means that the items can be incompatible to each other (\eg non-shareable), and thus cannot be in the same bin or incompatible to a bin (\eg logistic company not on competitor's platform AWS).

\subsection{Optimal placement model}

The subsequently defined optimal mixed integer linear programming (MILP) model for CEPP considers constraints $C1$--$C5$.
For the sake of clarity,~\cref{tab:variables} summarizes the notation used throughout this work in terms of constants, result variables and auxiliary variables.
\begin{table}[tb]
    \centering
    \begin{tabular}{p{0.7cm}|c|p{4.5cm}}
         Model & Variable & Meaning \\
         \hline
         \hline
		 Const. & $A_i$ & Required capacity of process $i$ \\
		 & $Q$ & Maximum number of processes per container \\
		 \hline
		 Result & $x_{ij}$ & Boolean indicating if process $i$ is in container $j$ \\
		 & $y_{nj}$ & Container $j$ is variant $n$ \\
		 \hline
		 Auxi-liary & $C$, $N$, $F$ & Number of: containers, container variants, processes \\
		 & $B_j$ & Capacity of container $j$ \\
		 & $G_j$ & Costs of container $j$ \\
		 & $H_j$ & Number of processes in container $j$ \\
    \end{tabular}
    \caption{Overview of used variables in the cost model}
    \label{tab:variables}
\end{table}

\labeltitle{Basic constraints} The first constraint (cf. C1) ensures that every process is placed in exactly one container by enforcing that only one of the variables $x_{ij}$ for each process takes the value 1, whereas $x_{ij}$ denotes that process $i$ is placed in container $j$.
	\begin{align*}
	\sum_{j=1}^{C}x_{ij} = 1 & & \forall i 
	\end{align*}
Then we ensure that for each container exactly one variant is chosen (cf. constraint C2).
This is accomplished by using a Boolean variable $y_{nj}$, which shows whether container $j$ is container variant~$n$.
	\begin{align*}
	\sum_{n=1}^{N}y_{nj} = 1 & & \forall j 
	\end{align*}
The used capacity (cf. C3) is calculated by multiplying every process size $A_i$ with the Boolean variable $x_{ij}$ and then summing up all products.
It is essential that the maximum capacity $B_j$ of container $j$ is respected.
	\begin{align*}
	\sum_{i=1}^{F}A_i \cdot x_{ij} \leq B_j & & \forall j 
	\end{align*}
To model the maximum number of processes per container (cf. C5), the number of processes in a container $H_j$, must be less than or equal to the given limit $Q$.
%
	\begin{align}
	&& H_j &\leq Q && \forall j &&\label{eq:maxFlowNumber} \\
	\text{With:} && H_j &= \sum_{i=1}^{F}x_{ij} && \forall j \label{eq:auxH}
	\end{align}
%

\labeltitle{Security-aware} The shareability property (cf. C4) is checked pair-wise processes ($i$ and $k$), and matches if at least one of them is non-shareable and they belong to different tenants.
In that case they must not be assigned to the same container, which is why $x_{ij}$ and $x_{kj}$ must not be $1$ at the same time for every container $j$.
The constraint is only applied for two processes ($i$ and $k$), if at least one of them is non-shareable and if they belong to different tenants.
	\begin{align*}
	x_{ij} + x_{kj} \leq 1 & & \forall i, j, k
	\end{align*}

\labeltitle{Objective}
Instead of determining which container is used and then setting its cost to $0$ via additional constraints (\ie increasing the complexity), an additional container variant of 0 size and 0 cost is added to the data set.
A container with this \enquote{zero-cost}-variant is a container which is not used, since no process can be assigned to it without exceeding its size of $0$, and thus it will not cause any costs just like an unused container.
As a consequence the overall costs can simply be calculated as sum of all container costs as shown in \cref{eq:AMilpObjective}.
	\begin{align}
	\min \sum_{j=1}^{C} G_j & &  \label{eq:AMilpObjective} 
	\end{align}

\section{Local search heuristic}
\label{sec:heuristic}
A CEPP solution of objective function \cref{eq:AMilpObjective} is optimal, but combinatorial optimization problems are NP-hard \cite{korte2012combinatorial}, and thus computationally intractable for practical problem sizes (\ie the better the solution, the more calculation time).
Hence, we define a local search heuristic tailored to CEPP that generates approximated results instead of an optimal solution (cf. objective (v)).
In general, local search uses one solution as a starting point and improves it iteratively by applying transformations (cf. \cite{korte2012combinatorial}).

Briefly, from an initial solution (\ie using a problem specific heuristic or random solution candidate), alterations / transformations to that solution are applied, if they improve the solution.
In this way, the solution's \emph{neighborhood} (\ie solutions reachable through one transformation) are explored by evaluating the new, improved solution using a \emph{measure} acceptance criteria), which is used to guide the search \cite{DBLP:series/anis/2010-10}.
The transformations are applied iteratively, until a specified \emph{termination criterion} is met (\eg empty neighborhood, time limit, fixed transformation limit, no improvements in last $k$ steps).
We consider a form of local search called hill climbing, which accepts only transformations leading to a better result.

\subsection{CEPP representation}
While in a MILP model Boolean variables are used to indicate which process is placed into which container ($x_{ij}$), we represent which process $i$ is assigned to which container $cId$ by a tuple $(i, cId)$, and which container $cId$ is assigned to which variant $vId$ by a tuple $(cId, vId)$.

\begin{figure}[tb]
	\centering
	\includegraphics[width=\columnwidth]{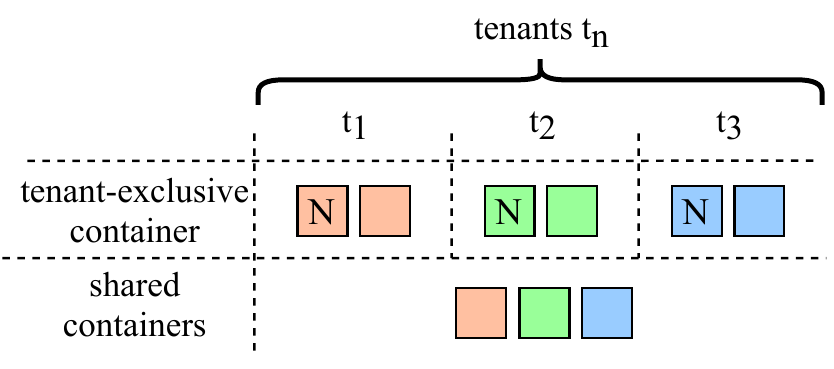} 
	\caption{Container types distinguished by their usage}
	\label{fig:containerTypes}
\end{figure}
Instead of constraints, we ensure that transformations result in feasible solutions, from a feasible initial solution.
For those transformations, we define \emph{shareable} (\ie only shareable processes) and \emph{tenant-exclusive} (\ie contain at least one non-shareable process of a tenant and thus can only contain processes of the same tenant) containers, as shown in \cref{fig:containerTypes}.
For simplicity, each container will be of one of these types without the ability to change the type during search.
\begin{figure}[bt]
	\centering
	\includegraphics[width=\columnwidth]{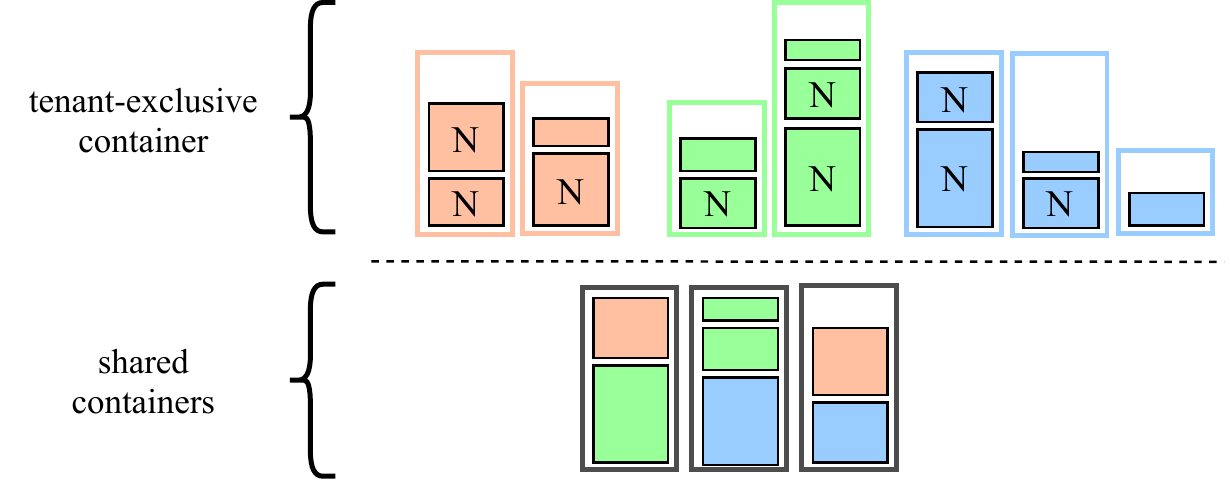} 
	\caption{Container types example}
	\label{fig:containerTypesExample}
\end{figure}
This is decided during creation of the initial solution and remembered by shareable $shCId \coloneqq set(cId)$ and tenant-exclusive container identifers $tenantExCId \coloneqq Tenant \times set(cId)$.
\Cref{fig:containerTypesExample} shows a CEPP example solution, where containers are marked tenant-specific according to their type (red, green and blue), with shareable containers outlined grey color, since they are not tenant-specific.

The resulting total costs are calculated by adding up the costs for all container variants, and thus represent the objective function similar to \cref{eq:AMilpObjective}.
Unused containers are not excluded, but an artificial variant is added to the list of possible variants, which has 0 costs and capacity (cf. similar to \cref{sec:milp}).

\subsection{Initial CEPP solution --- bin packing}
Instead of generating a random solution, we use a problem-specific heuristic, called first fit decreasing (FFD) \cite{korte2012combinatorial}.
FFD is a bin packing algorithm that first sorts the items descending by their size.
If the current item does not fit into a bin, a new one is created, while repeating until all items are distributed. 
We choose FFD for its unmatched runtime and its good, fixed bound for the maximum deviation $FFD(I) \leq \frac{11}{9} \cdot OPT(I) + \frac{2}{3} \label{eq:FFD}$ (cf. \cite{korte2012combinatorial,zhang1997new}).

In particular, we set out to calculate an initial, however, not optimal solution of the CEPP with a modified FFD, where a bin represents a container and a process denotes an item, which we subsequently motivate and describe in more detail.

\subsubsection{FFD-CEPP problems} Since FFD is designed to determine the number of bins for only one size it does not need a strategy for choosing the bin's size when adding a bin.
However, the CEPP is a variable size bin packing problem, \ie with bins of different sizes (cf. \cite{correia2008solving,kang2003algorithms,wascher2007improved}).
Consequently, FFD has to be extended by a rule must be specified that defines the bin to choose.

Moreover, in normal bin packing the only constraint is denoted by the size of a bin that must not be exceeded by the packed items.
But when conflicts between items are added to the problem, a distribution strategy based on FFD might not find a solution.
Therefore FFD must be additionally adapted for bin packing with conflicts (cf. \cite{jansen1999approximation,sadykov2013bin}).
Subsequently, we specify adaptations of FFD to variable size bin packing and bin packing with conflicts.

\subsubsection{FFD for variable size bin packing} The only requirement when adding a new bin is that the current item fits into the bin.
This can be achieved by always choosing the biggest bin available, which would violate the goal of cost minimization.
Hence, we sort the bins by their cost efficiency (\ie the higher the costs per capacity unit, \eg Euro/MB, the less efficient is the container), and not their sizes. 
When a new bin has to be initialized, its type will be chosen as follows: the bins are tested in the order of their efficiency, starting with the most efficient bin.
If the bin is big enough to host the item it is chosen, else the next bin is tested.

Since it is not known how many items are still going to be distributed when adding a new bin, we can never be sure how \enquote{full} the container will get.
However, when assuming that the container will not have much free capacity left at the end of the bin packing,
the cost efficiency indicates which bins should be preferred.
The more cost-efficient the bin, the less the costs at the end will be.

\begin{figure}[tb]
	\centering
	\includegraphics[width=\columnwidth]{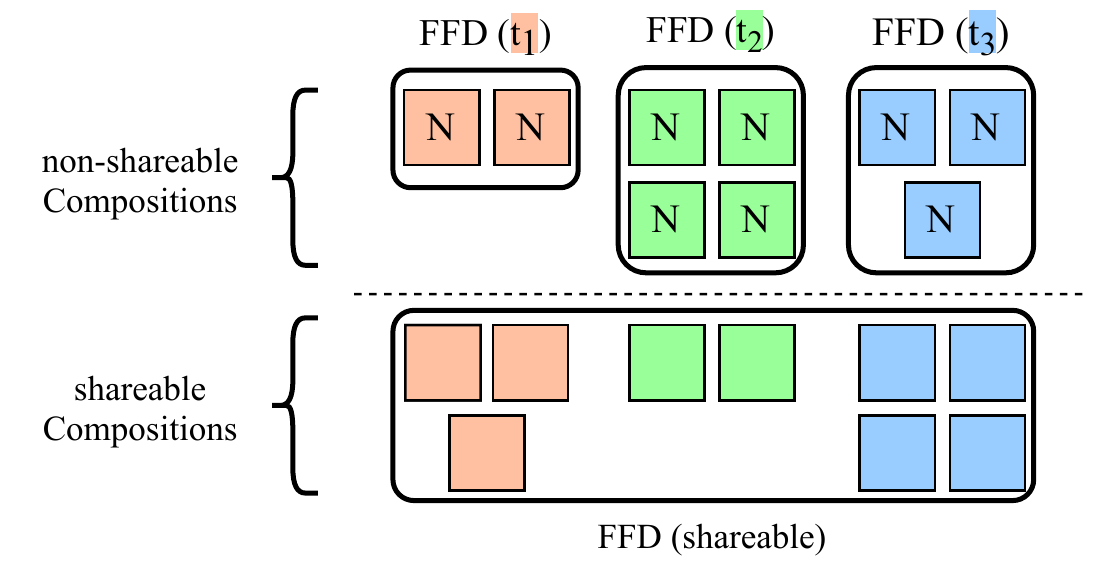} 
	\caption{Initial solution for process placement using FFD}
	\label{fig:initialSolutionWithFFD}
\end{figure}
\begin{figure*}[tb]
\begin{subfigure}{\columnwidth}
  \centering
  \includegraphics[width=1\columnwidth]{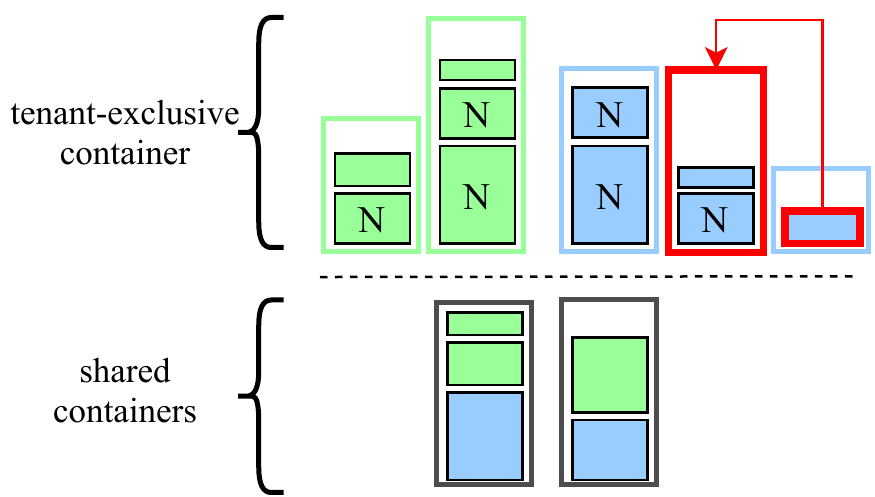} 
  \caption{Pre-condition}
\end{subfigure}%
\begin{subfigure}{\columnwidth}
  \centering
  \includegraphics[width=1\columnwidth]{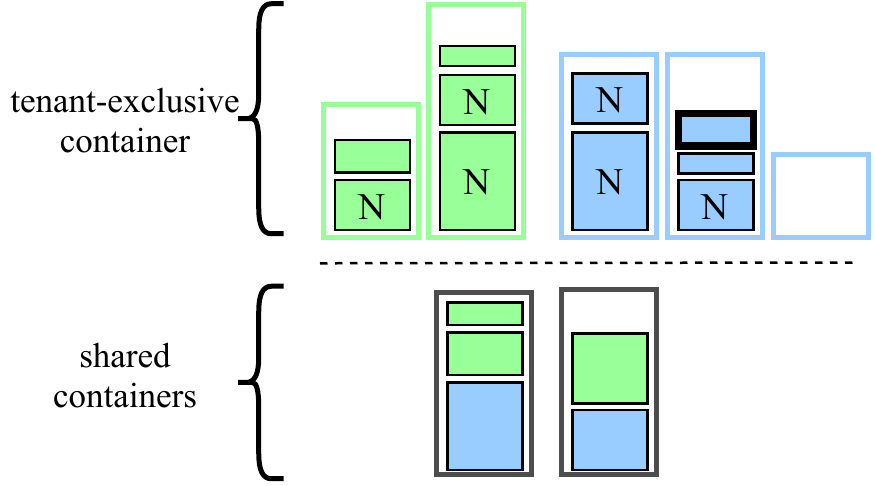} 
  \caption{Post-condition}
\end{subfigure}
\caption{Moving an process to another container}
\label{fig:transformMove}
\end{figure*}
\begin{figure*}[tb]
\centering
\begin{subfigure}{\columnwidth}
  \centering
  \includegraphics[width=1\columnwidth]{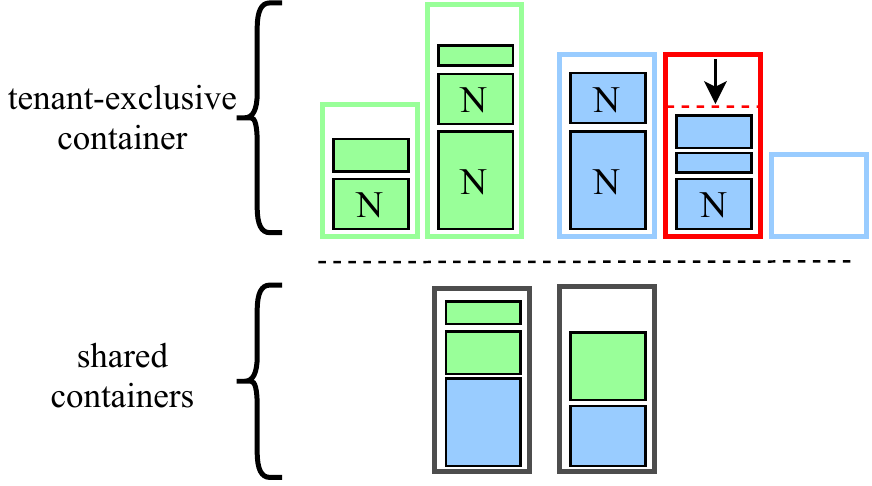} 
  \caption{Pre-condition}
\end{subfigure}%
\begin{subfigure}{\columnwidth}
  \centering
  \includegraphics[width=1\columnwidth]{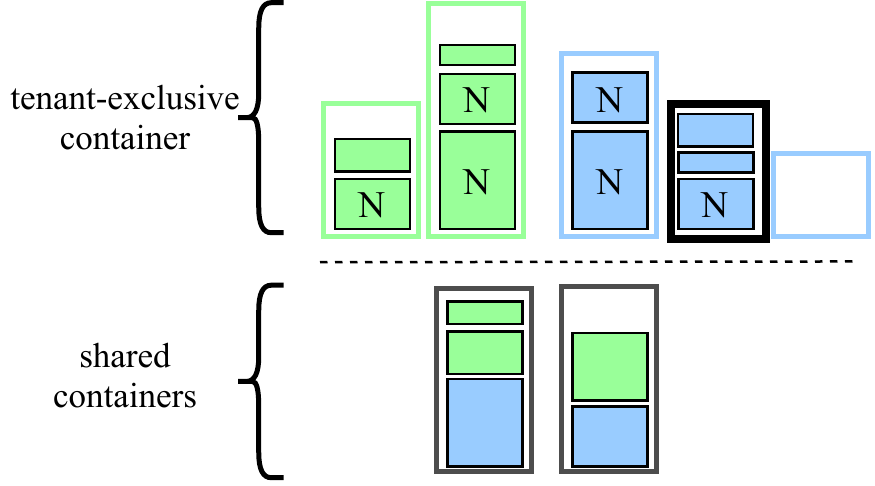} 
  \caption{Post-condition}
\end{subfigure}
\caption{Exchanging a container variant with a smaller variant}
\label{fig:transformShrink}
\end{figure*}
\begin{figure*}[bt]
\centering
\begin{subfigure}{1\columnwidth}
  \centering
  \includegraphics[width=1\columnwidth]{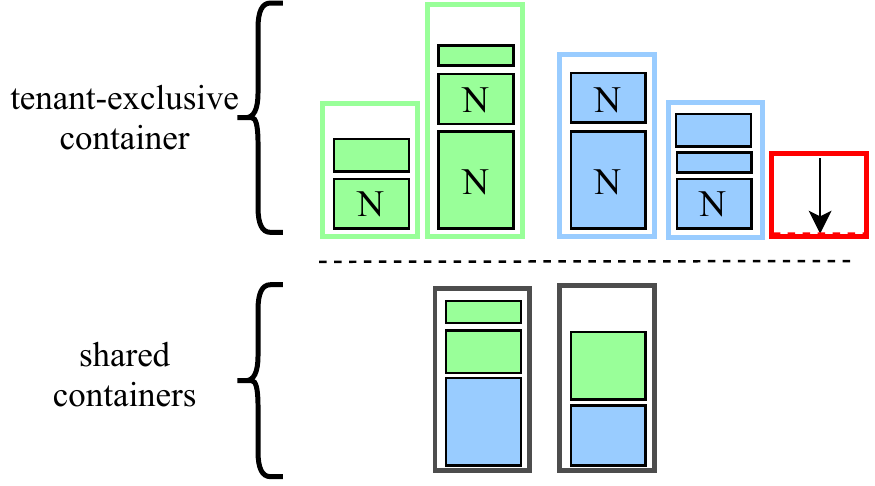} 
  \caption{Pre-condition}
\end{subfigure}%
\begin{subfigure}{1\columnwidth}
  \centering
  \includegraphics[width=1\columnwidth]{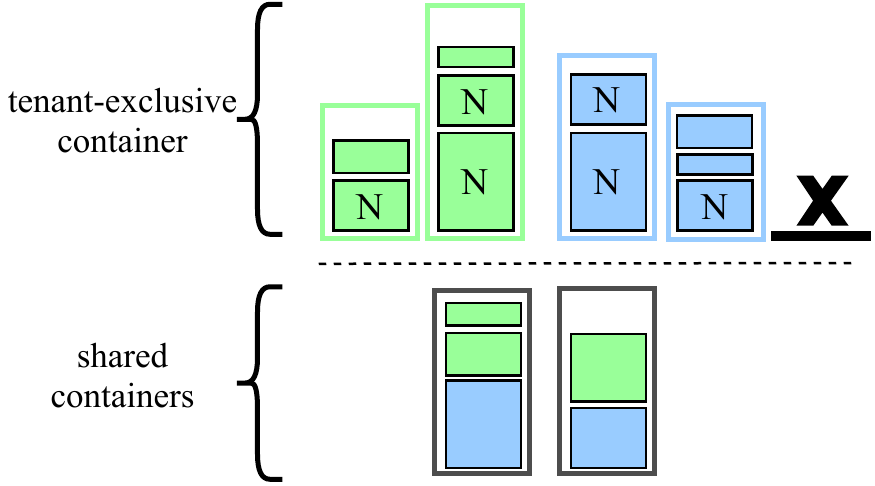} 
  \caption{Post-condition}
\end{subfigure}
\caption{Removing a container by assigning a variant of 0 size and costs}
\label{fig:transformRemove}
\end{figure*}
\begin{figure*}[bt]
\centering
\begin{subfigure}{1\columnwidth}
  \centering
  \includegraphics[width=1\columnwidth]{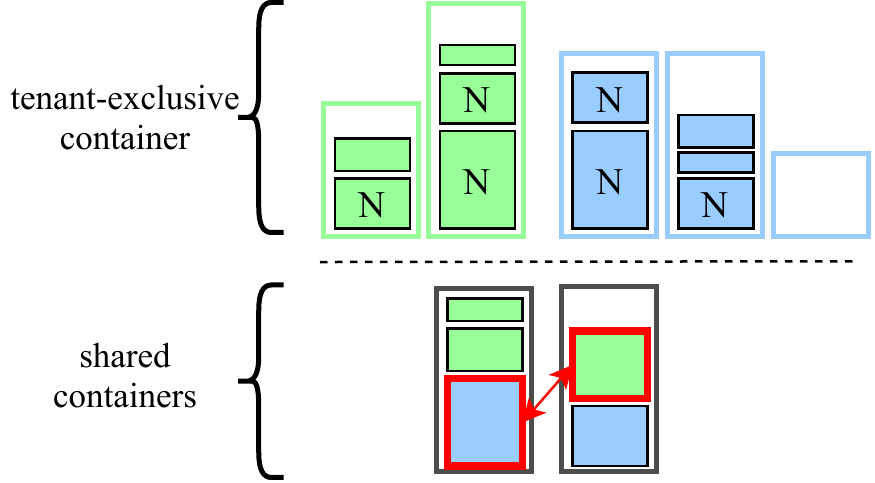}
  \caption{Pre-condition}
\end{subfigure}%
\begin{subfigure}{1\columnwidth}
  \centering
  \includegraphics[width=1\columnwidth]{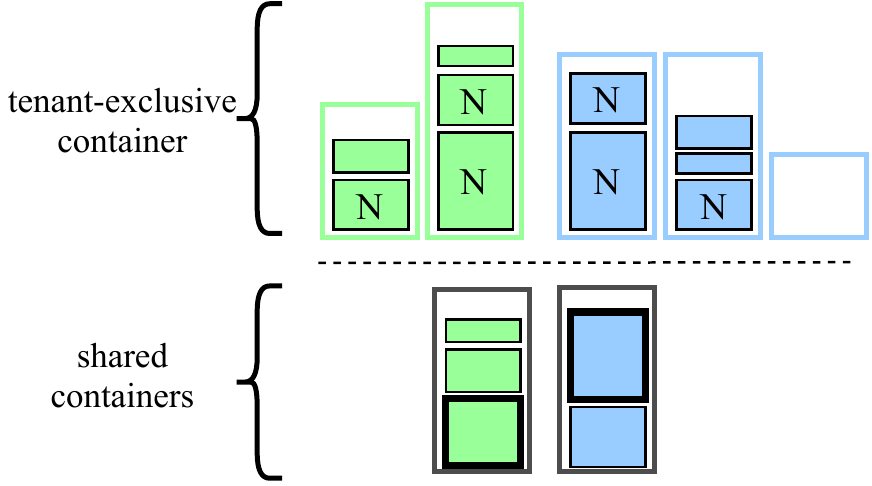} 
  \caption{Post-condition}
\end{subfigure}
\caption{Swapping a container}
\label{fig:transformSwap}
\end{figure*}

\subsubsection{FFD for bin packing with conflicts}
When processes cannot be assigned to the same bin due to a conflict (\ie non-shareable processes in combination with an process from a different tenant), we choose a strategy based on problem decomposition (\eg cf. \cite{korte2012combinatorial}) that does not require to test each container whether a conflict occurs (and choosing the next one).
We group all non-shareable processes per tenant and all shareable processes as separate sub-problems, which can also be applied when creating the initial result for the local search.
When using our variable size FFD to the sub-problems separately, no conflicts can occur within the sub-problems, and thus our modified FFD can be used to generate a feasible initial solution.
For example, \cref{fig:initialSolutionWithFFD} shows how FFD is applied to non-shareable IPCG, grouped by tenant.
The resulting placement into containers is depicted by rectangles and stored to $shCId$, $tenantExCId$, respectively. 

\subsection{CEPP transformations}

For CEPP we specify transformations \emph{shrink} and \emph{reduce} that rearrange container variants with the goal to use cheaper variants.
While \emph{shrink} exchanges a bin for a smaller and cheaper bin, \emph{reduce} removes an unused one.
As a second type of transformation on processes we define \emph{move} that rearranges the process placement by moving an process to a different container.
Notice that simply applying reordering transformations on processes alone would not reduce the overall costs.
Finally, we introduce a \emph{swap} transformation that exchanges two processes from different containers that results into a better compaction within containers.

Subsequently, we introduce four transformations as pre- and post-condition pairs that first show the intended change as precondition, before the transformation is applied, and the desired post-condition as the result, after a viable transformation was accepted.
The transformations guarantee conflict-free solutions by moving processes only to those container types ($shCId$ or $tenantExCId$), which are allowed according to their properties.
Since it cannot be guaranteed that the container size is not exceeded, this is checked by the acceptance criteria.

\subsubsection{Move} To rearrange IPGCs between containers, one random process is chosen, and then possible containers are determined, to which it can be moved, depending on its shareability characteristic.
If the process is non-shareable, only containers from $tenantExCId$ are allowed, else shareable containers from $shCId$ can be additionally used.
From the possible containers, one is randomly chosen and the process is assigned to it.
A prerequisite for moving an process is that the chosen container has enough free space, else, the move is discarded by the acceptance criteria.
\Cref{fig:transformMove} shows an example of a move transformation.

\subsubsection{Shrink} In this transformation we randomly select a container and change its variant.
The possible variants are ordered by their size, assuming that most expensive variants were removed beforehand.
The transformation looks up the variant identifier $vId$ assigned to a container in $(cId, vId)$ and reduces the number by 1.
Due to the ordered variants, a smaller variant identifier means a smaller and cheaper variant, and thus reducing the total cost.
If the capacity of the new container would be exceeded, the transformation is revoked.
This transformation cannot be applied to the smallest container variants (\ie with $vId == 0$).
An example of container shrinking is shown in \cref{fig:transformShrink}.

\subsubsection{Remove} If a solution (incl. initial solution) generates too many containers some of them should be removed.
A removed container is the same as a container which is not used.
Since an additional variant with 0 size and costs was added (\ie \enquote{shrink to zero}), this transformation is applied implicitly when a container is shrunk to this variant or a container is moved, leaving a container empty.
\Cref{fig:transformRemove} shows the removal of a container.

\subsubsection{Swap} To exchange two processes between different containers, two processes are chosen randomly, and moved to the other container, and vice versa.
The prerequisites for exchanging two processes are (a) that either both processes are shareable (\ie currently in \emph{shCId}) or both their containers are in \emph{tenantExCId}, and (b) that both containers have enough free space, else, the swap is discarded by the acceptance criteria.
\Cref{fig:transformSwap} shows an example of a swap.

\subsubsection{CEPP local search algorithm} We combine the previous specifications into \cref{alg:myLocalSearch}, which denotes our CEPP local search algorithm.
\begin{algorithm}[tb]
 \caption{CEPP local search heuristic}
 \label{alg:myLocalSearch}
 Generate initial solution: CEPP--FFD\;
 \While{Maximum number of transformations not exceeded}{%
  \ForEach{transformation $\in$ (Move, Swap, Shrink)}{%
  	Apply transformation\;
  \eIf{Solution is feasible and better or equal to original solution}{%
   Accept transformation\;
  }{%
	Discard transformation\;  
  }
 }
 }
\end{algorithm}
At first the initial result is generated using our modified FFD algorithm separately for non-shareable processes per tenant and all shareable processes (line 1).
Then move, swap and shrink transformations are iteratively applied (in list order) to the initial solution, which is modified with the goal to reduce costs (line 4).
In case a transformations is not applicable, it does not change the original solution.
A transformation is accepted if the modified solution is feasible and in addition to has lower or equal costs (line 6).
The search is terminated when all transformations were applied.

\section{Evaluation}
\label{sec:evaluation}
In this section we evaluate our CEPP-MILP and local search-based solutions with respect to cost-efficiency in multiclouds and performance.
Additionally, we conduct a case study on cost-aware process modeling leveraging the CEPP solution.
First we discuss potential cost-savings in case studies by example of two real-world scenario categories from SAP CPI \cite{sap-hci-content}, and then we conduct a quantitative study on the latency and solution quality of the local-search approach compared to the optimal solution.
Finally, we combine our solution for CEPP with process cost improvements (\ie combine pattern and router to routing slip rewriting) to study cost aware modeling regarding its viability and their relevance for real-world scenarios, as well as the interaction with the modeler.

\subsection{Setup}

We briefly introduce common aspects of our experiments, including a cost model for the valuation of the capacity $\text{CAP}$ (cf. \cref{def:characteristics}).

\labeltitle{Cost models} For all experiments we consider available cost models for container variants from the leading platform vendors Microsoft Azure and Amazon AWS, as well as one less-known, anonymous vendors (\ie X), as shown in \cref{fig:variantsWithoutOverlyExpensive}.
It can be seen that for most sizes at least two providers have an container variant offering.
One advantage of multicloud is too choose the container according to your requirements.
Since for CEPP only the two features of a container variant shown in the diagram (size and costs) are relevant, more expensive offerings (from other vendors) with the same sizes (\eg RAM) can be neglected, and thus reduced to $14$ variants (not shown), which also reduces the search space.
Although we acknowledge that there are other relevant resource capacities like CPU time, we exclusively use RAM measured in MB to distinguish container variants during our evaluation (mainly because we get RAM usages from benchmarks like \cite{ritter2016benchmarking}).
Our general approach, however, can be easily used with arbitrary resource capacities.

The AWS cost model seems to attract customers with cheaper, smaller container sizes (\ie break-even at around 8GB RAM), while becoming more expensive for more bigger container variants (\eg for $16$GB and $64$GB RAM the costs are slightly higher than for Azure and X).
For the latter, Azure and X seem to be preferable.
\begin{figure}
  \centering
  \includegraphics[width=\columnwidth]{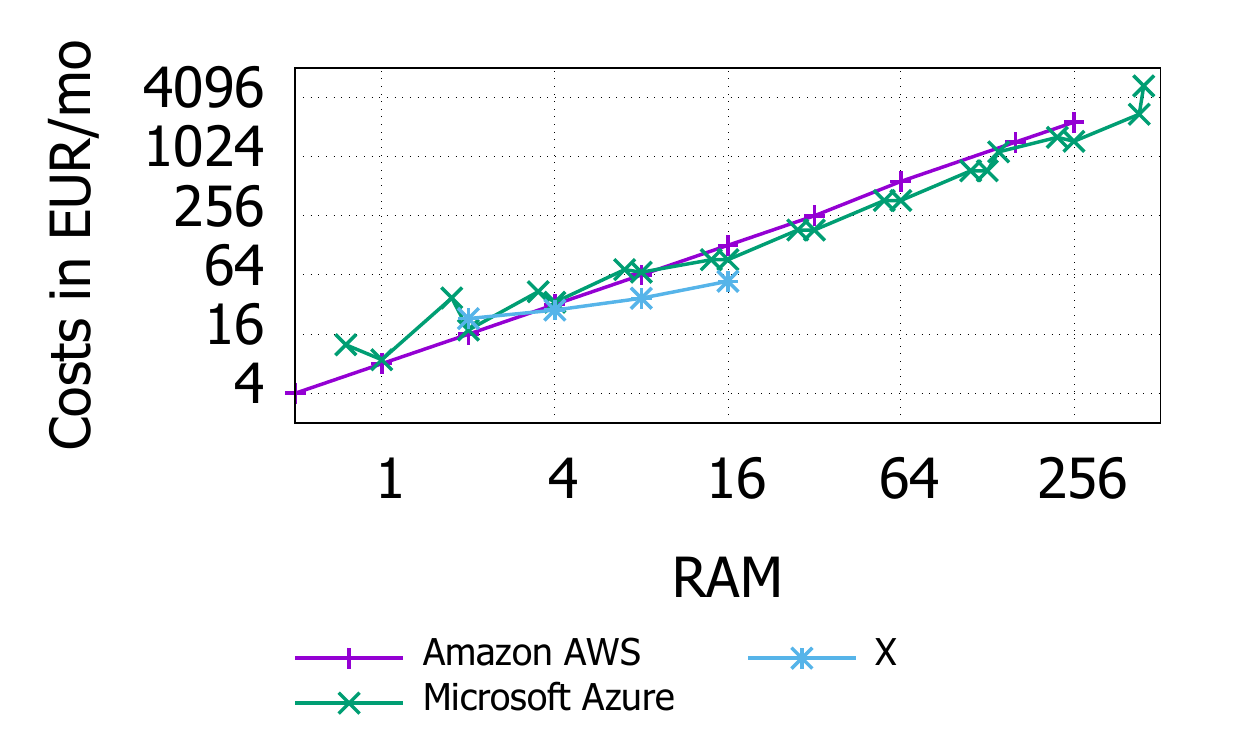}
  \caption{Costs per hoster container variant}
  \label{fig:variantsWithoutOverlyExpensive}
\end{figure}

\labeltitle{Process decomposition --- \enquote{cutting} sub-processes} As first indicated in \cref{fig:composition_partitioning} and further discussed in \cref{sec:decomposition}, bigger processes could be cut into smaller sub-processes that might be placed more efficiently (cf. \enquote{component decomposition} in \cite{DBLP:conf/europar/SilvaP17}).
This type of decomposition of processes into smaller shareable and non-shareable sub-processes is only considered as part of the first case studies.

All measurements are conducted on a Intel Core i5-6300U CPU with 2.40GHz 2.5Ghz, 16.0 GB RAM, Windows 10 operating system, and JDK 1.8. 

\subsection{Case studies --- is placement across multiclouds beneficial?}
\label{sub:motivating_example}
Before we investigate whether a placement across multiclouds is beneficial, we revisit our motivating example.
When applying our approach to the motivating example (cf. \cref{ex:motivation}), the optimal solution is a placement of the processes to two containers from different vendors (AWS (30,00 EUR/mo), X (20,00 EUR/mo)) with cost $50.00$ EUR/mo, resulting to cost savings between 10.00 and 30.00 EUR/mo.
Already this simple example indicates that sharing processes across different platforms (\ie multicloud) can help to reduce costs for the INTaaS vendor.
Now, let us look into two more complex, real-world scenarios.

\begin{figure*}[bt]
	\centering
	\includegraphics[width=.7\linewidth]{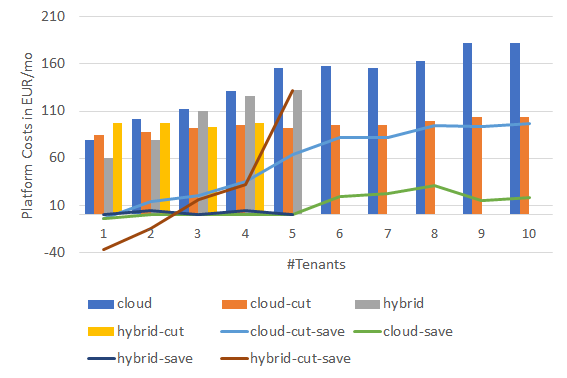}
	\caption{Costs and savings: original and cut}
	\label{fig:costs_savings}
\end{figure*}
\subsubsection{SAP Cloud Platform Integration: eDocuments}
The eDocument Electronic Invoicing is a solution for country-specific document management \cite{sap-hci-content}, allowing firms to interact with legal authorities. 

\paragraph{Setup} The dataset for country solutions like Italy and Spain contains $13$ distinct processes (six with non-shareable sub-processes), which we study with a uniform distribution over an increasing number of mock tenants, \ie representing fictitious tenants (due to confidentiality).

\paragraph{Results} The costs and savings for the calculated partitions compared to a hosting solution are shown in~\cref{fig:costs_savings} as \texttt{cloud}.
Notably, \enquote{cutting} non-shareable sub-processes (cf. \texttt{cloud-cut}) into $19$ distinct processes is not beneficial in the one-tenant case, but preferable for more tenants (cf. cost-saving for cloud \texttt{cloud-save} and cut cloud \texttt{cloud-cut-save} processes).
Working with the original processes becomes profitable from five tenants onwards.
The solver latencies for original and cut processes vary between 500ms and 2.7s, respectively.
The solution for the 19 \enquote{cut} processes and three tenants costs $112.32$ EUR/mo (cf.~\cref{fig:cloud_pre-partitioned}) for a separation to three vendor platform variants and identifies free capacity.

\paragraph{Conclusions} (1) a multicloud placement shows cost-savings especially for more complex cloud integration scenarios; (2)  \enquote{cutting} is beneficial from a certain number of tenants onwards; (3) small problem sizes are still computationally tractable.

\begin{figure*}[tb]
\centering
\begin{subfigure}{.7\linewidth}
  \centering
  \includegraphics[width=.8\linewidth]{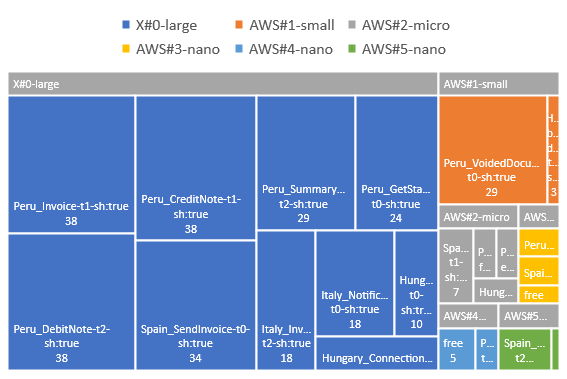} 
  \caption{eDocuments (costs=112.32 EUR/mo)}
  \label{fig:cloud_pre-partitioned}
\end{subfigure}%
\\
\begin{subfigure}{.7\linewidth}
  \centering
  \includegraphics[width=.8\linewidth]{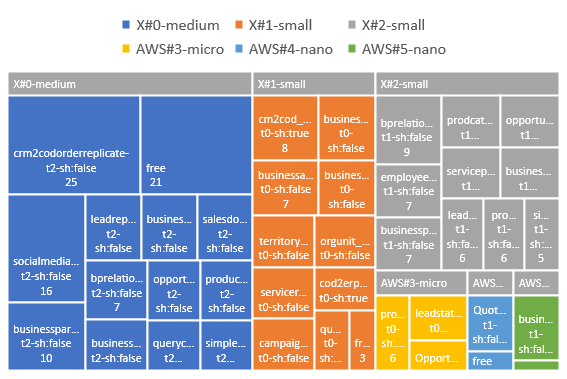} 
  \caption{C4C-CRM (costs=109.85 EUR/mo)}
  \label{fig:c4c_normal_partitioning}
\end{subfigure}
\caption{SAP CPI eDocuments and C4C-CRM processes from \cite{sap-hci-content}, $tx$ tenants with $x \in\{0,1,2\}$ (anonymized)}
\end{figure*}
\subsubsection{SAP Cloud Platform Integration: C4C-CRM}
The C4C solutions for the communication with on-premise Customer Relationship Management (CRM) applications~\cite{sap-hci-content}, abbreviated \emph{c4c-crm}~\cite{sap-hci-content}, can be considered a typical hybrid, corporate to cloud application integration~\cite{Ritter201736}.
The dominant integration styles --- according to the classification in~\cite{Ritter201736} --- are process invocation and data movement.
The state changes (\eg create, update) of business objects (\eg business partner, opportunity, activity) as well as master data in the cloud or corporate applications (\eg CRM, ERP) are exchanged with each other.

\paragraph{Setup} The dataset features $37$ distinct processes (seven with non-shareable sub-processes), which we uniformly distributed over an increasing number of mock tenants denoted as \texttt{hybrid}.

\paragraph{Results} The cost-savings for $78$ \enquote{hybrid-cut} processes, in \cref{fig:costs_savings}, make configurations with more than two tenants, five tenants compared to c4c-crm more profitable (cf. cost-saving for cut hybrid processes \texttt{hybrid-cut-save}).
The uncut hybrid processes do not show much cost-saving potential (cf. \texttt{hybrid-save}).
This can be explained due the higher number of processes after decomposition, but with a better shareablity ratio.
On the downside the solver times reach up to 2h for the latter, which makes it not applicable in practice (\ie no entries beyond five tenants).
The optimal solution for the $37$ original processes and three tenants is shown in~\cref{fig:c4c_normal_partitioning}.
The selected six platform variants from two vendors cost 109.85 EUR/mo.
Due to vendor preferences (all processes on one platform) for the non-cut processes, this solution is close to a hosting approach.
While the transmission latencies of these processes remain low, there are only minor cost-savings.
Without \enquote{pinned} platform assignments and cut processes, the cost-savings are high, but with higher latencies due to an additional remote call.

\paragraph{Conclusions} (4) already slightly higher numbers of processes are computationally intractable (\ie on-line heuristics required); (5) trade-off between latency and cost-savings is an open research question.

\subsection{Quantitative analysis --- local search heuristic}

As discussed before and also seen during the case studies (cf. conclusion (4)), larger problem sizes are intractable for the optimal CEPP-MILP approach.
For example, with more than $170,000$ processes and more than 30,000 tenants as in SAP CPI \cite{sap-hci-content}, a heuristic approach is required.
To assess our heuristic we discuss processing times and solution quality on an anynomized version of exactly this CPI data set (incl. original tenant assignments).

\subsubsection{Performance --- Latency} First we study the processing latency on a comparable subset of the 170k process data set, taking into account that the optimal solution cannot process larger problem sizes. 

\paragraph{Setup} Since we want to compare the optimal solution with our heuristic, the subset of processes for which we measured the performance step-wise for \{5, 10, 20, 50, 100\} processes, until we did not receive a result for the MILP-based solution in a certain period of time.

\paragraph{Results} \Cref{fig:latency} shows the results of the step-wise performance measurements for the FFD-heuristic with $1,000$ (\ie FFD-1k) and $10,000$ transformation iterations (\ie FFD-10k) and the process data \{all-5, ..., all-100\} for $5$, .., $100$ processes, as discussed.
\begin{figure}
  \centering
  \includegraphics[width=1\columnwidth]{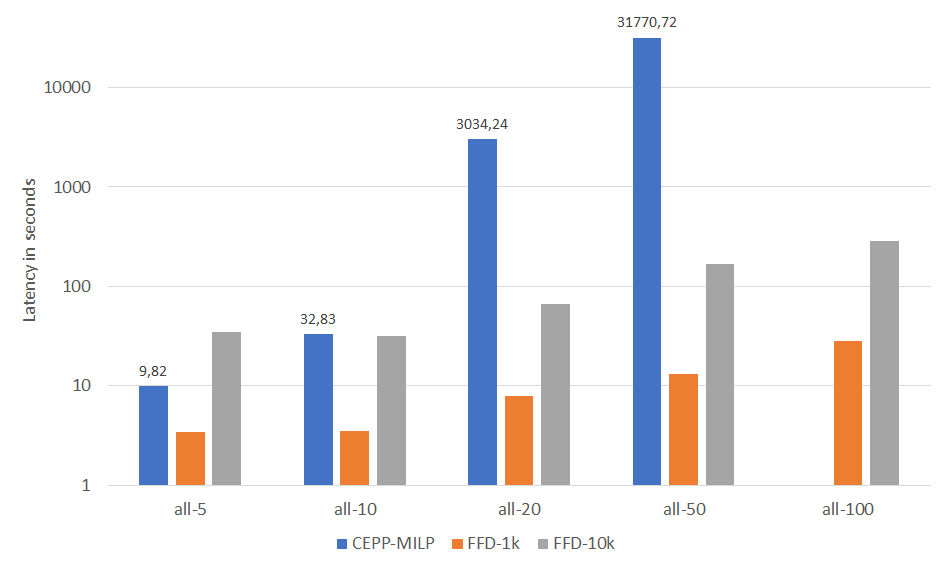} 
  \caption{Latency: modified FFD heuristic}
  \label{fig:latency}
\end{figure}
It can be seen that almost for every problem instance the performance was increased by more than 90\%.
With an increasing number of transformations, the latency increases modestly.
In contast, the MILP approach shows enormous increases in solving times (none for all-100).
For very small instances (\eg five IPGGs), 10k transformations even leads to a performance decrease, compared to the optimal solution.
MILP-solver could not calculate a result for $100$ processes within $12$ hours, leading to an end of the experiment.

\paragraph{Conclusions} (6) the modified FDD heuristic can be applied in practical INTaaS cases; (7) for small problem sizes, the MILP-approach is competitive.

\subsubsection{Solution quality} While the performance results are promising, we evaluate the quality of the solutions of the heuristic approach, compared to the optimal solution (cf. \cref{eq:AMilpObjective}).

\paragraph{Setup} same setup as for latency experiments.

\paragraph{Results} \Cref{fig:quality} shows the solution quality of the heuristic compared to the optimal solution up to $50$ processes.
While the solution quality of the heuristic on the cloud and hybrid data sets from our case study are close to optimal, \ie showing deviations of less than 10\% (not shown), the deviations increase to values between 30\% and 40\%.
\begin{figure}
  \centering
  \includegraphics[width=1\columnwidth]{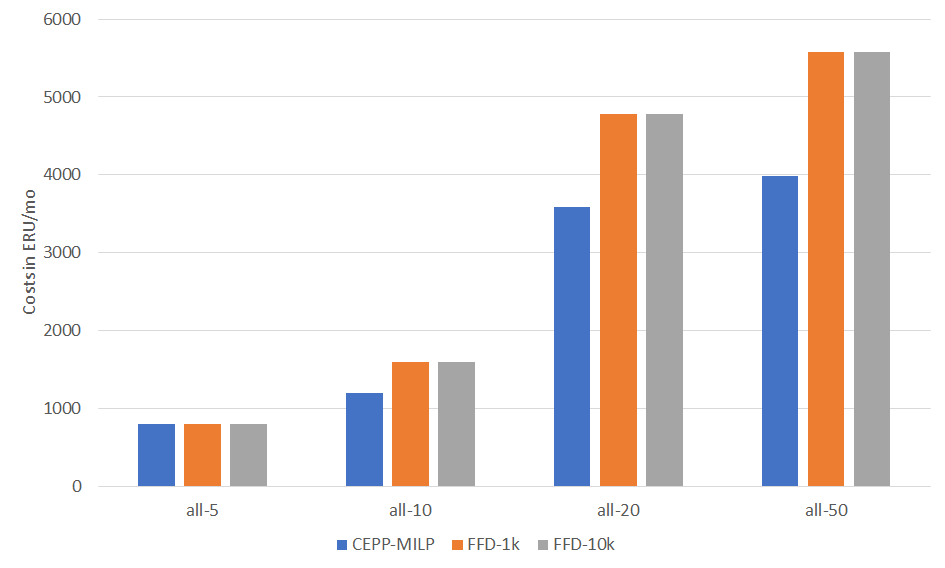} 
  \caption{Solution quality: modified FFD heuristic}
  \label{fig:quality}
\end{figure}
Solely, for a small number of processes, the heuristic result is optimal.

\paragraph{Conclusions} (8) solution quality of the heuristic is approximately off by $\frac{1}{3}$.

\subsection{Putting cost-aware modeling into practice}

While the results from the case studies show potential cost savings from the perspective of the INTaaS providers (considering their customers' requirements), they do not allow to answer questions of process modelers like \enquote{What does a process cost?} and \enquote{Is this process cost-efficient?} (without development and maintenance costs).
These questions might occur during the development of the processes, which could allow the user who creates the process to budget the development.
However, when recalling the cost-aware modeling process in \cref{fig:cost_aware_modeling_process}, the interaction with the modeler requires \emph{viability} regarding quick round-trip times between the modeling step and the combined correctness, cost calculation change application iterations, leading to the change proposal step, and \emph{relevance} regarding actual modeling improvement potentials.
Hence, we first quantitatively study viability and relevance on a real-world setup from SAP CPI \cite{sap-hci-content}, and then discuss cost-aware modeling in the context of our motivating example (cf. \cref{fig:cdd}).

\subsubsection{Viability --- modeler interaction}
The presented heuristic allows for finding solutions of practical problem sizes like the $170,000$ processes and more than $30,000$ tenants \cite{sap-hci-content} within half an hour.
While this allows INTaaS vendors to use our solution one could argue that even shorter round-trips are required to give the modeler timely feedback or change proposals with improved costs.
These round-trip times are largely dominated by the cost calculation step in \cref{fig:cost_aware_modeling_process}.

We recall that the structural correctness is built-in to the modeling (cf. \cref{sec:decomposition}), and thus does not take any extra time.
For semantic correctness, model checking is applied only to the current process, however, could take a significant amount of time (\eg state explosion problem \cite{clarke2011model}).
Since we do not consider semantic correctness in this work, we assume that no checks are performed.
Further, applying the decomposition (cf. \cref{sub:decomposition_rules}) and model improvement rules (cf. \cref{sub:cost_rewriting_rules}) is local to the process and consequently dependent on the size of the process as well as the time to converge (\ie no more rule applicable or no changes to the model anymore), leaving the cost calculation heuristic as the bottleneck.
However, since processes will run \enquote{locally} in certain regions, from a modeler's perspective, not all cloud regions in a INTaaS multicloud setup are relevant and would need to be checked regarding costs.
For example, the process in \cref{fig:cdd} will exclusively run in Italy or more generally in the \emph{European} region.
In fact, the costs for this particular process could be calculated in a subset of the complete aforementioned real-world problem size.
Hence, we subsequently analyze the anonymized problem sizes of SAP CPI \cite{sap-hci-content} by region to give a more realistic impression into the viability of our solution.

\paragraph{Setup} We group telemetry data from SAP CPI (\eg tenants, processes) by geographical region, each combining several local data centers.

\paragraph{Results} \Cref{fig:metabase} shows load shares on (anynomized) regions.
\begin{figure}
  \centering
  \includegraphics[width=.6\columnwidth]{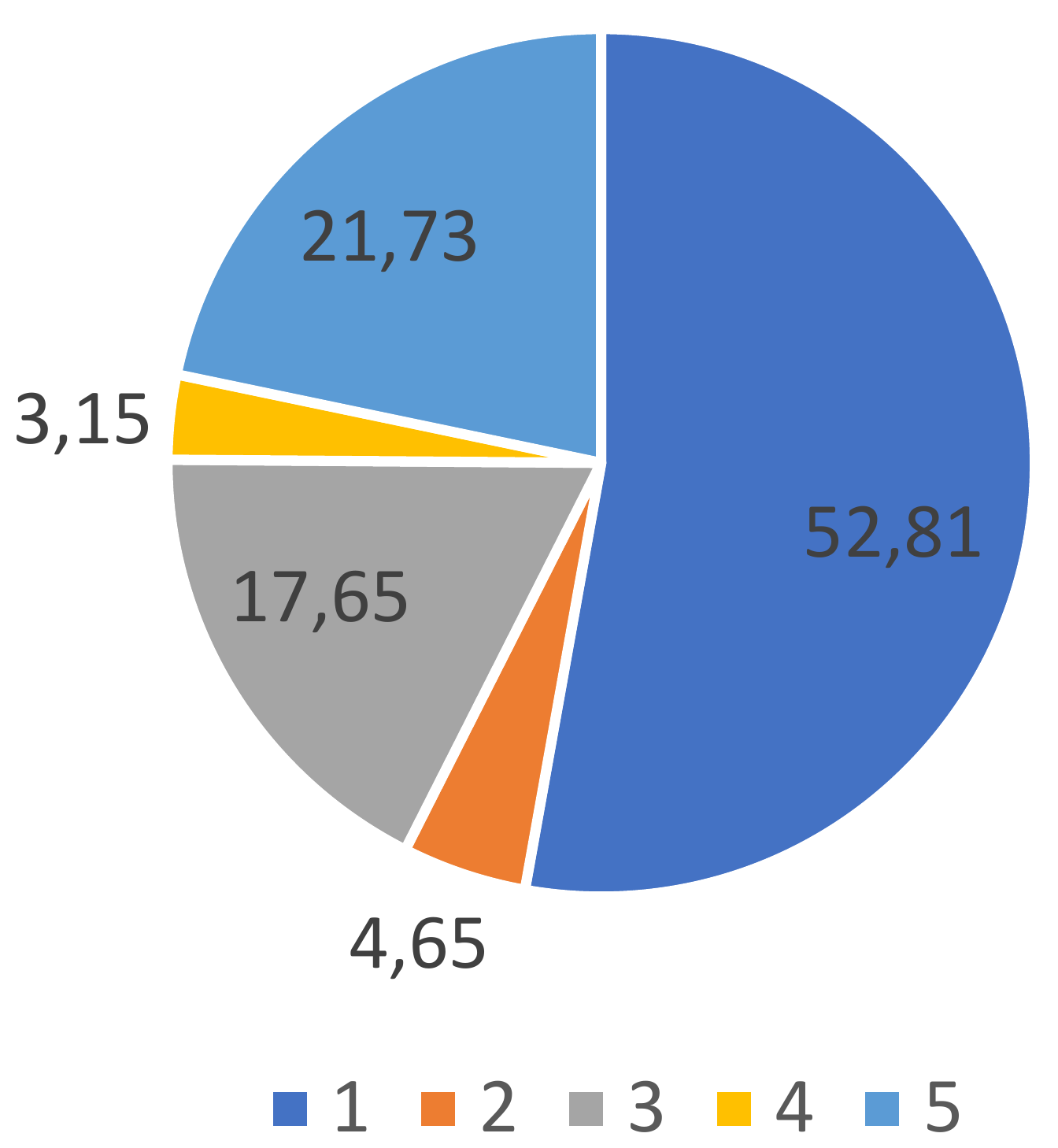} 
  \caption{Load shares on regions in \%, with each region having several data centers (anynomized)}
  \label{fig:metabase}
\end{figure}
Notably, limiting calculations to geographic locations could significantly reduce the problem size.
While an optimal solution can still not been applied to smaller regions (\eg an optimal cost calculation for $200$ processes takes approximately $12$h), the heuristic could find solutions in most regions within few seconds for the 10k and even faster for the 1k transformation configuration.
For bigger regions and 10k transformations, solutions can be found ranging from few minutes up to 30 minutes for the complete problem size (all regions).
Hence, an application of our heuristic partitioned by geographic regions, reduces the round-trip times and makes the interaction with the modeler viable.

\paragraph{Conclusions} (9) model checking semantic correctness is highly desirable to extend the process correctness step in our cost-aware modeling process (cf. \cref{fig:cost_aware_modeling_process}), but could have a significant impact on performance (cf. future work); (10) cost calculations limited to geographic regions reduce modeling round-trip times.

\subsubsection{Relevance --- modeling improvements}
In \cref{tab:optimization_strategies} we assessed common integration process improvements regarding their impact on the reduction of the operation costs of an INTaaS.
These improvements are applied in the \emph{Apply changes} step in the cost-aware modeling process in \cref{fig:cost_aware_modeling_process} and proposed as more cost-efficient process model changes to the user.
While the structural correctness of these changes is guaranteed by definition (cf. \cref{sec:decomposition}), we have yet to assess their relevance in a real-world INTaaS setup.

Hence, similar to \cite{DBLP:conf/debs/0001MFR18} we consider integration processes from SAP CPI 2017 standard content (called \texttt{ds17}), and compare it with scenarios from 2015 (called \texttt{ds15}).
The comparison with a previous content version features a practical study on content evolution.
To analyze the difference between disjoint scenario domains, we grouped the scenarios into the following categories~\cite{Ritter201736}:
On-Premise to Cloud (\texttt{OP2C}),
Cloud to Cloud (\texttt{C2C}), and
Business to Business (\texttt{B2B}).
Since hybrid integration scenarios such as OP2C target the extension or synchronization of business data objects, they are usually less complex.
In contrast native cloud application scenarios such as C2C or B2B mediate between several endpoints, and thus involve more complex integration logic~\cite{Ritter201736}.
The process catalog also contained a small number of simple Device to Cloud scenarios; none of them could be improved by our approach.

Subsequently, we study the impact of optimization strategies OS-1 and OS-2 from \cref{sec:decomposition} regarding their resource reduction potential (\ie relevance).

\labeltitle{Reduce static resource consumption: Process simplification (OS-1)}
Since resources like CPU usage and main memory are part of all major platform cost models, the required runtime software stack and its configuration \emph{statically} contributes to the overall costs.
For example, non-shareable patterns like script or message translator require configurations like library dependencies, instantiated (mapping) programs and bundled schema files that can easily reach a high amount of allocated main memory.
Unnecessary pattern instance configurations that are never used by any process can result in stale, but memory consuming installations.
In a first study, we analyse the effect of process simplifications on the static resource consumption.

\paragraph{Setup} The integration processes are stored as process models in a BPMN-like notation~\cite{ritter2016exception} (similar to \cref{fig:cdd}).
The process models reference data specifications such as schemas (\eg XSD, WSDL), mapping programs, selectors (\eg XPath) and configuration files.
For every pattern used in the process models, runtime statistics are available from benchmarks~\cite{ritter2016benchmarking}.
The data specifications are picked up from aforementioned data sets ds15 and ds17, while the runtime benchmarks are collected as used in SAP CPI.
The mapping and schema information is automatically mined and added to the patterns as contracts, and the rest of the collected data as pattern characteristics.
For each integration process and each optimization strategy, we determine if the strategy applies, and if so, if there is an improvement.
This analysis runs in about two minutes in total for all $902$ scenarios on our workstation.

\paragraph{Results} The relevant metric for the process simplification strategies from OS-1 is the model complexity, \ie the average number of pattern reductions per scenario, shown in \cref{fig:disconnected}.
\begin{figure}
  \centering
  \includegraphics[width=\columnwidth]{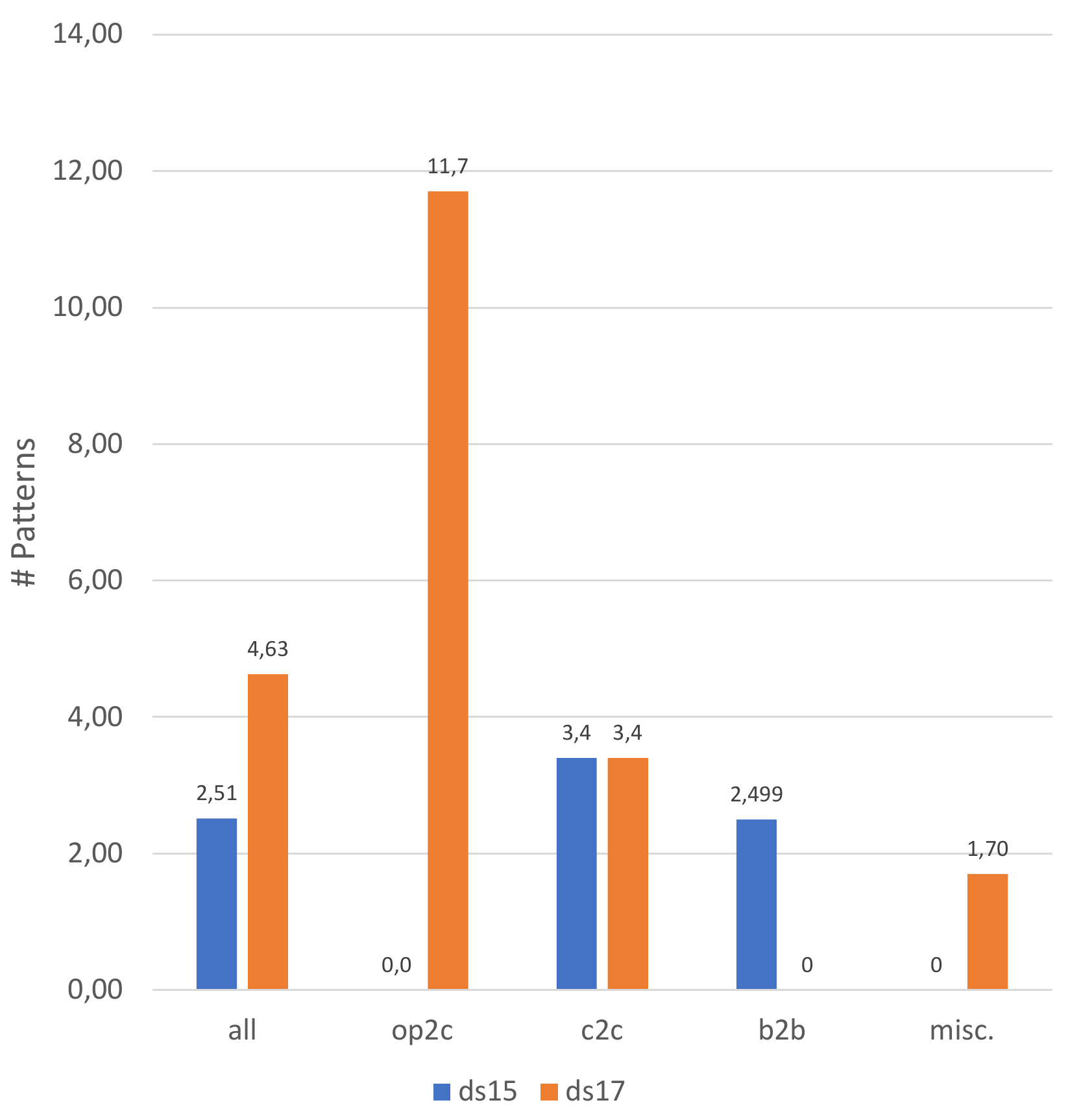} 
  \caption{Pattern removal potential}
  \label{fig:disconnected}
\end{figure}
Although all scenarios were implemented by integration experts, familiar with the modeling notation and the underlying runtime semantics, there is still a small amount of patterns per scenario that could be removed without changing the execution semantics.
On average, the content reduction for the content from 2015 and 2017 was $2.51$ and $4.63$ patterns/IPCG, respectively, with significantly higher numbers in the OP2C domain.

\paragraph{Conclusions} (11) Even simple process simplifications are not always obvious to integration experts in scenarios represented in a control-flow-centric notation (\eg current SAP CPI does not use BPMN Data Objects to visualize the data flow) and the need for process simplification does not seem to diminish as integration experts gain more experience; and (12) the amount of statically allocated main memory could be reduced for most of the categories.

\labeltitle{Reduce workload-specific resources: Data Reduction (OS-2)}
Data reduction is workload-specific and has the potential to impact memory consumption and bandwidth costs.
To evaluate data reduction strategies from OS-2, we leverage the data element information attached to the processes' IPCG contracts and characteristics, and follow their usages along edges in the graph, similar to \enquote{ray tracing} algorithms~\cite{glassner1989introduction}.
We collect the data elements that are used or not used, where possible --- we do not have sufficient design time data to do this for user defined functions or some of the message construction patterns, such as request-reply.
Based on the resulting data element usages, we calculate unused data elements in \cref{fig:savings}.

\paragraph{Setup} same as before.

\paragraph{Results} There is a large amount of unused data elements per scenario for the OP2C scenarios; these are mainly web service communication and message mappings, for which most of the data flow can be reconstructed.
\begin{figure}[bt]
  \centering
  \includegraphics[width=\columnwidth]{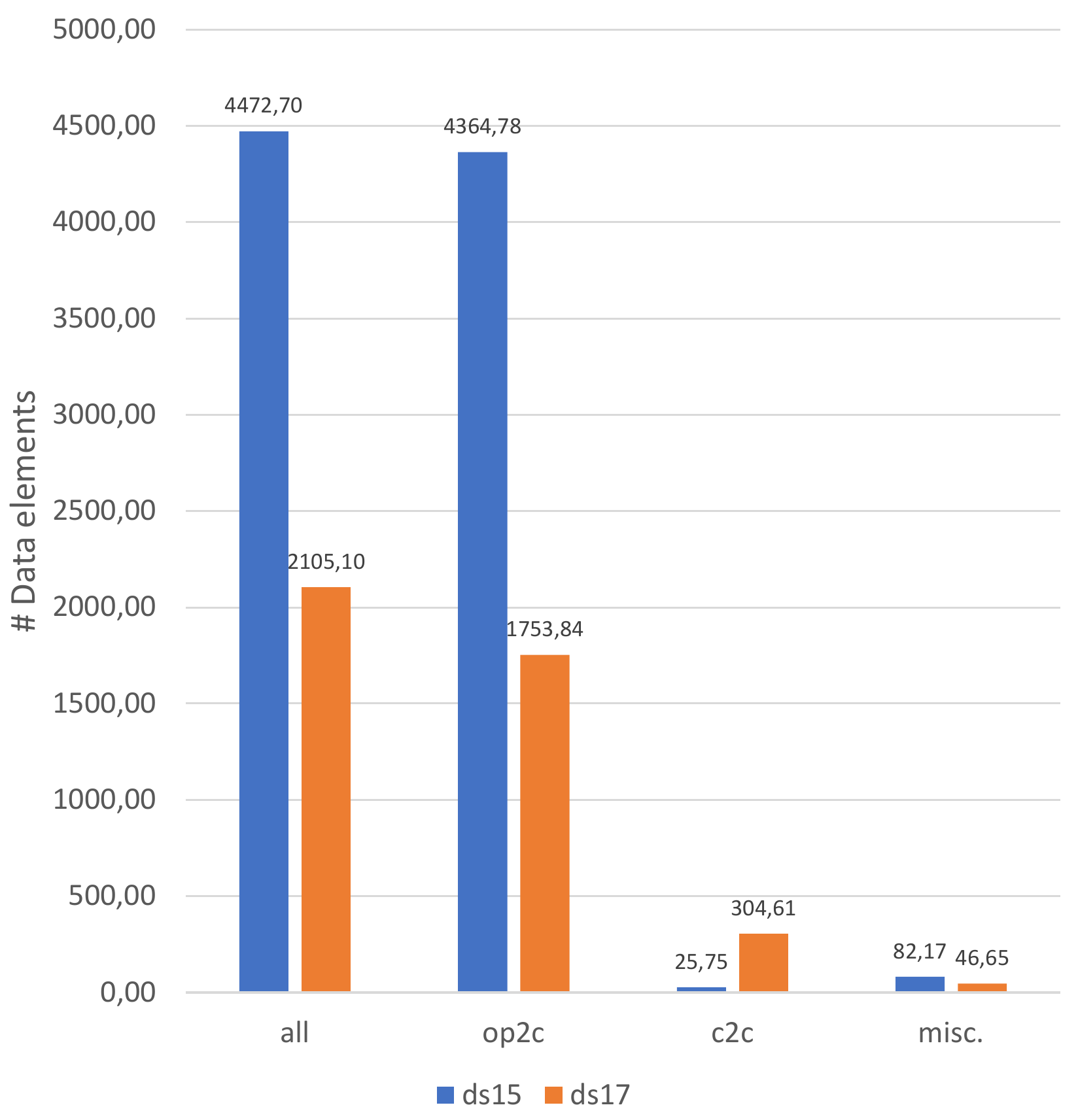} 
  \caption{Data element removal potential}
  \label{fig:savings}
\end{figure}
This is because the predominantly used EDI and SOA interfaces (\eg SAP IDOC, SOAP) for interoperable communication with on-premise applications define a large set of data structures and elements, which are not required by the cloud applications, and vice versa.
In contrast, C2C scenarios are usually more complex, and mostly use user defined functions to transform data, which means that only a limited analysis of the data element usage is possible.
Overall there is an immense memory reduction potential, especially for OP2C scenarios.

\paragraph{Conclusions} (13) Data flows can best be reconstructed when design time data based on interoperability standards is available; and (14) a high number of unused data elements per scenario indicates cost reduction potential.

\subsubsection{Cost-aware modeling --- Italy invoicing}
After having assessed the practical viability of the underlying cost calculation and the relevance of process simplification and data reduction improvements, we study the interaction between INTaaS vendor and process modeler by revisiting our motivating example, while following the cost-aware modeling process in \cref{fig:cost_aware_modeling_process} in three phases.

\paragraph{Phase 1 (user $\rightarrow$ INTaaS): Modeling, decomposition and correctness} A user / modeler implements the Italy invoicing process from \cref{fig:cdd}.
The resulting process model is transformed into an IPCG, on which the INTaaS platform can assess the structural (and semantic) correctness of the process model.
Due to our definitions in \cref{sub:compositions}, the structural correctness could already be enforced during modeling or simply checked thereafter (\eg when the process model is saved and the IPCG is constructed).

During the first phase, the process model can be decomposed using the graph rewriting rules introduced in \cref{sec:decomposition}.
The resulting separated shareable and non-shareable processes would be correct according to our definitions, but the invoicing example does not have patterns with side-effects\footnote{The \emph{store accessor} pattern in SAP CPI accesses a store local to the process instance that is tenant-separated and not configurable, and thus has no side-effects.}.

\paragraph{Phase 2 (INTaaS): Cost calculation and change enumeration} If correct, the formalized integration processes represented as processes are transformed into our CEPP model (cf. \cref{sec:milp}).
Together with the already existing processes in CEPP representation from one or many geographical regions as well as the platform vendor cost models, the costs of placing the processes is calculated, allowing for an answer on a per process level (cf. \cref{ex:invoicing_multicloud}).
\begin{example}
    \label{ex:invoicing_multicloud}
    The invoicing process could be placed in a multicloud setup on an \texttt{AWS t2.small} with costs of approximately $15.94$ EUR/mo\footnote{Considering one tenant without platform vendor preferences and a main memory consumption of $64$MB per pattern.}.
    Already this could improve the costs compared to a tenant-specific, custom hosting approach with approximate costs of 72.53 EUR/mo, indicating that a multicloud placement is beneficial. \qedblack
\end{example}
Notably, the costs measured could include the actual costs from the container variant of the platform vendor with additions from the INTaaS provider (incl. margin, operation and maintenance).

Thereafter, improvements from \cref{sec:decomposition} are applied in the form of an optimizing compiler and costs are calculated, generating several change proposals for the process with their cost reduction potential.
The cost of the current process as well as the change proposals are then shown to the modeler in the use process modeling tool.

\paragraph{Phase 3 (INTaaS $\rightarrow$ user): Groom proposed changes} A change proposal, shown to the modeler, could be the one in \cref{ex:invoicing_reworked}, which is essentially improved using our new process simplification optimizations: \emph{combine patterns}, \emph{router to routing slip}.
That removes two unnecessary process step / pattern instances as well as several more, when smartly making the routing implicit.
This denotes a non-trivial, multicloud cost-reduction potential that a user could not have anticipated easily (\eg the respective MS Azure container would result in higher costs).
At the same time, it significantly reduces the model's complexity, which makes it better understandable.
\begin{example}
    \label{ex:invoicing_reworked}
    The rewriting rules from \cref{sec:decomposition} are applied to the process in \cref{fig:cdd} and results to the cost improved process depicted in \cref{fig:cdd_improved}.
    \begin{figure*}[bt]
    	\centering
    	\includegraphics[width=0.8\linewidth]{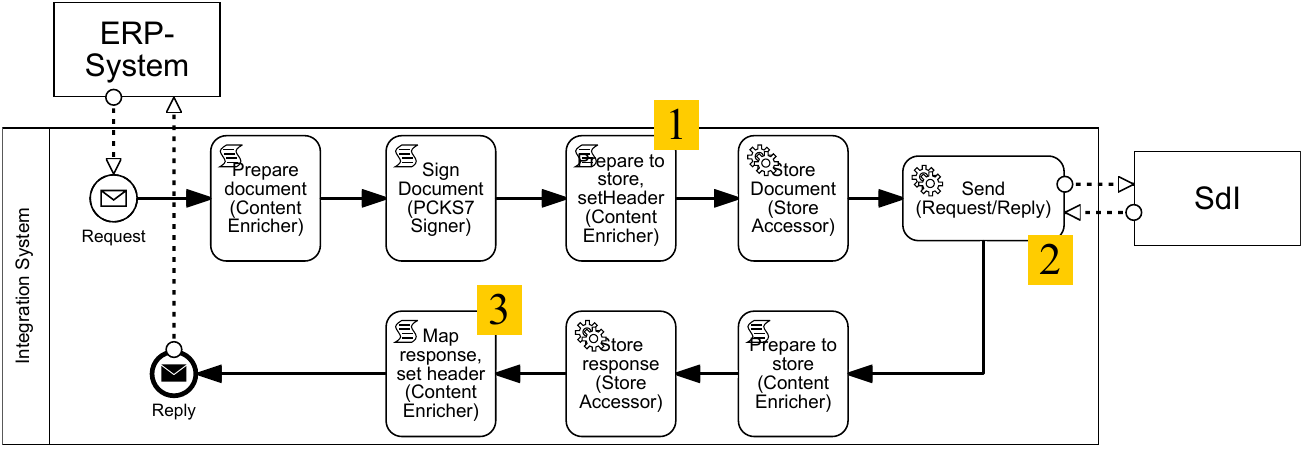}
    	\caption{SAP CPI eDocuments Italy invoicing integration process from \cref{fig:cdd} (more cost-efficient)}
    	\label{fig:cdd_improved}
    \end{figure*}
    First, \textbf{{\colorbox{amber!70}1}} the content enricher instances \texttt{Prepare to store} and \texttt{setHeader} are merged by combining their configurations (\eg expressions).
    This is allowed, since their changes are not overlapping with the \texttt{PCKS7 signer} pattern.
    Similarly, \textbf{{\colorbox{amber!70}3}} the content enricher instances \texttt{map response} and \texttt{setHeader} can be merged.
    In \textbf{{\colorbox{amber!70}2}}, the explicit differentiation of \texttt{Sdl\_Test} and \texttt{Sdl\_Production} in the process, represented by a Content-based \texttt{Router} and two \texttt{Request-Reply} patterns could be moved to an implicit configuration.
    This allows to remove the router as well as the additional \texttt{End} (cf. edge \texttt{Invalid Mode}) and one \texttt{Request-Reply} pattern.
    The routing conditions on \texttt{mode} of the modeled content-based router pattern are mapped to implicit routing configurations as \emph{endpoint URL} in the form of a content enricher, essentially forming a routing slip pattern \cite{hohpe2004enterprise}, which is subsequently merged into the newly created combined \texttt{Map response, set header} step.
    These changes reduce the process by five patterns and make it now fit into an \texttt{AWS t2.micro}, which would reduce the costs to $7.97$ EUR/mo.
    \qedblack
\end{example}

\paragraph{Conclusions} (15) current development of integration processes does not consider their costs (\ie cost-aware modeling support required); (16) a study of the improvements vs. costs trade-off (\ie cost-impact of optimization strategies) is missing.

\subsection{Discussion}

The optimal MILP-based solution performs very well on smaller problem sizes, but  becomes increasingly intractable with a growing problem size (cf. conclusions (3)-(4), (7)).
In contrast, the heuristic performs very well (cf. conclusion (6)) and can calculate a solution for the $170,000$ processes within $30$ minutes, but leading to a trade-off of practical applicability vs. cost-savings, with a solution quality decrease of the heuristic by $\frac{1}{3}$ (cf. conclusion (8)).
Additionally, there is the trade-off of latency vs. costs, \eg when decomposing processes and distributing the sub-processes to different platforms (cf. conclusion (5)).
This trade-off also applies, when several interacting processes are placed to different vendor platforms. 

For cost-aware modeling we found that a heuristic approach combined with geographically scoped cost calculations could make a solution viable in the context of real-world problem sizes (cf. conclusion (10)).
In addition, process change proposals could provide relevant cost reductions (cf. conclusions (12)--(14)), and could help to guide a user during the non-trivial task of cost-aware process modeling (cf. conclusions (11), (15)).
However, the evaluation helped to identify several interesting challenges for further investigations:
\begin{enumerate*}[label=(\alph*)]
    \item While a formal analysis of the execution semantics would be desirable for semantic correctness of a process\footnote{The semantic correctness of a process is crucial for a practical application of our work. 
    However, it essentially requires a formal notion of correctness of the IPCG's underlying execution semantics, which is beyond the scope of this work.
    Briefly, in prior work \cite{DBLP:conf/edoc/0001RMRS18,ritter2019formal,DBLP:phd/at/Ritter2019} we have shown that the execution semantics of IPCGs can be given by a variant of \tdbnets with boundaries, for which the correctness of the graph transformations in \cref{sec:decomposition} would have to be shown (\eg through bi-simulation equivalence of the original vs. changed process).}, they would potentially lead to a viability trade-off with respect to its evaluation times (cf. conclusion (9)).
    \item There could be other ways of scoping cost calculations than geographical regions (\eg platform vendor preferences).
    However, all scoping approaches might eventually have to deal with issues like contradicting cost measurements in different scopes (\eg for a change proposal).
    \item Acceptance of process change proposals might be subject to company policies (\eg auditing, automatic acceptance).
    \item Pricing policies and effects between platform and INTaaS vendors as well as their customers should be further studied (\eg propagation of prices, re-evaluation of modeling decisions, monitoring and alerting).
    \item Besides studies concerning the interaction between user / customer and INTaaS, more extensive studies on the cost-impact of optimization strategies are necessary, \eg improvements vs. costs trade-off (cf. conclusion (16)).
\end{enumerate*}

In summary, cost-aware modeling offers immense cost saving potentials for INTaaS vendors and their customers, but come with further questions and interesting trade-offs that can now be studied, for the first time.

\section{Related Work}
\label{sec:relatedwork}
\begin{table*}[]
\begin{tabular}{p{2cm}|p{2.2cm} p{2.0cm} p{2cm} p{2cm} p{1.5cm} p{2.2cm}}
Literature / Objectives & Automatic multicloud (i) & Correctness (ii) & Cost-efficient (iii) & Perfor-mance (v) & Security (iv) & User perspec-tive (vi) \\
\hline 
\hline
Applications       & \faThumbsOUp & \faThumbsODown & \faThumbsUp & \faThumbsUp & \faThumbsUp & \faThumbsODown \\ 
Service            & \faThumbsUp & \faThumbsODown & \faThumbsUp & \faThumbsUp & \faThumbsODown & \faThumbsODown \\ 
Modeling   & \faThumbsODown & \faThumbsOUp & \faThumbsODown & \faThumbsODown & \faThumbsODown & \faThumbsODown \\ 
Miscellaneous      & \faThumbsUp & \faThumbsODown & \faThumbsODown & \faThumbsODown & \faThumbsOUp & \faThumbsODown \\ 
\hline
CPEE/CAPM               & \faThumbsUp & \faThumbsOUp & \faThumbsUp & \faThumbsUp & \faThumbsUp & \faThumbsUp \\
\hline
\end{tabular}
\begin{tablenotes}
	\centering
	\small
	\item \faThumbsUp: provides solution, \faThumbsOUp: partial solution, \faThumbsODown: no solution
\end{tablenotes}
\caption{Related work in the context of the objectives of cost-aware process modeling (CAPM)}
\label{tab:relatedWork}
\end{table*}
In the literature the CEPP can be seen as a combination of the variable size bin packing (VSBPP)~\cite{correia2008solving,kang2003algorithms,wascher2007improved}, where bins have different capacities and costs with the goal of minimizing the total cost associated with the bins (commonly solved using a solution for the cutting-stock problem \cite{gilmore1961linear} and local search heuristics), and the bin packing with conflicts problem~\cite{jansen1999approximation,sadykov2013bin}.
The only work we found that is slightly similar to our approach is by Epstein et al. \cite{EPSTEIN2011333}.
Although the latter does not consider costs or preferences (\eg security, avoid platform vendors), it essentially proposes a more efficient solution of the First-Fit-Decreasing (FFD) algorithm for the VSBPP~\cite{zhang1997new}, which we modify regarding variable size problems and apply in a certain way to avoid conflicting bin assignments (cf. \cref{sec:heuristic}).
Subsequently, we discuss related work in other domains, set into context to our objectives (i)--(vi) and our CEPP approach in \cref{tab:relatedWork}.
Notice that although some domains seem to cover most of the objectives, this results from the combined literature in the domain.

\labeltitle{Application placement / database placement}
Closest known related work is that on distributed applications by Silva et al. \cite{DBLP:conf/europar/SilvaP17,DBLP:conf/ccgrid/SilvaPD16}, considering performance and cost-efficiency when distributing application components in multi-clouds environments.

In the area of secure application or program placement, \eg SWIFT, JIF by Chong et al.~\cite{chong2007secure} requires manual annotation of confidentiality and integrity labels in the JIF language for an \enquote{automatic} assignment, while we automatically determine shareability labels derived from the pattern specifications.

Our work is inspired by Bollwein et al. \cite{bollwein2017separation} who propose splitting multiple database tables across multicloud hosted databases providers using integer linear programing (\ie splitting data not processes).

\labeltitle{Service placement} A related problem in multicloud systems is service provisioning, which is approached by game theory (\eg~\cite{anuradha2014survey,ardagna2017generalized,passacantando2016service}) with a focus on data throughput, considering connected systems in non-cooperative games for resource allocation, but not vendor costs and security constraints.

Additionally, Legillon et al. \cite{DBLP:conf/cec/LegillonMRT13} add the cost-perspective to identify a cost-efficient service deployment in multiclouds.
For a practical usage, automatic deployment operators are identified (\eg \enquote{scind mutation}, \ie splits one machine into two smaller ones, and \enquote{scale down mutation}, \ie exchanges one machine by a smaller one, would be used in the proposed placement in \cref{ex:motivation}).
Objectives like security are not considered.

\labeltitle{Cloud process modeling} We found early work on modeling workflow interactions in cloud enviroments from Zhou et al. \cite{zhou2011scalable} who study liveness properties based on the Petri net grounded semantic constraint net formalism.
Although the Petri net grounding allows for formal analysis, the main focus is on the synthesis of collaborating workflows, but not considering cloud or multicloud aspects (\eg costs, security) or the user perspective.

In the BPM domain, Wei\ss enberg et al. \cite{weissenberg2014cloud} propose a logistics modeling approach in the cloud with the objectives of cost savings and efficient processes.
However, the work merely mentions the cost saving potentials of the cloud with respect to on-demand or pay-per-use cost models and does not provide a solution in the context of our work.

\labeltitle{Miscellaneous}
Similarly for efficient resource allocation in cloud data centers for allocating virtual machines in multi-tenant setups (\eg~\cite{li2015efficient}), and virtual machine placement (\eg \cite{simarro2011dynamic}).
More closely related is the work on privacy levels for business logic in multicloud deployments~\cite{nacer2017metric}.

\section{Conclusions} 
\label{sec:conclusion}

In this work, we provide a solution for the cost-efficient process placement problem for integration processes in multicloud and studied its applicability to real-world INTaaS integration processes ($\rightarrow Q1$).
We conclude that SaaS / INTaaS vendors can cut costs, when following our approach (conclusions (1)+(2)).
The complexity of the CEPP makes an optimal solution intractable (cf. conclusions (3)+(4), (6)+(7), and leads to inevitable trade-offs (conclusions (5)+(8)).
Still, our local search heuristics allow for more timely calculations in the context of cost-aware modeling (conclusion (10)), making our contributions applicable to real-world multicloud INTaaS ($\rightarrow Q2$), which we studied up to a problem size of 170,000 processes and more than 30,000 tenants.
The cost-aware modeling process could help an INTaaS' customer to reduce costs ($\rightarrow Q3$) and simplify their integration process models and configurations (conclusions (11)--(15)), while assuring their structural correctness ($\rightarrow Q2$).

Future work will be conducted especially in the context of semantic correctness and the identified correctness vs. viability trade-off (conclusion (9)), for the practical applicability vs. cost-savings (\eg considering improvements of the heuristic like better initial solutions, more sophisticated transformations) and latency vs. costs (\eg through extension of our model by communication latencies for geographically distributed containers).
The interaction between the user and INTaaS perspectives as well as the interplay between process cost improvements through change proposals and overall cost heuristics denote further interesting fields of study (conclusion (16)).

\noindent\textbf{Acknowledgements} We thank Jasmin Greissl for various valuable discussions in the context of this article and Dr. Fredrik Nordvall Forsberg for proofreading. 

\bibliographystyle{elsarticle-num}
\bibliography{partitioning}

\end{document}